\documentclass[11pt,a4paper]{article}
\usepackage[left=2cm, right=2cm, top=2cm]{geometry}

\usepackage{float} % To re-style the 'figure' float
\usepackage{graphicx}
\usepackage{amsmath}
\usepackage{mathtools}
\usepackage{amssymb}
\usepackage{bm}
\usepackage{framed}
\usepackage{subfig}
\usepackage{hyperref}
\usepackage{nomencl}
\usepackage{color}
\usepackage{textcomp}
\usepackage{graphics}
\usepackage{epsfig}
\usepackage{epstopdf}
\usepackage{amsthm}
\usepackage[linesnumbered,ruled,vlined]{algorithm2e}

\usepackage{latexsym,amsfonts}
\usepackage{url}
\usepackage{longtable}
\usepackage[figuresright]{rotating}
\usepackage{listings}
\usepackage{algpseudocode}
\usepackage{footnote}
\usepackage{xcolor}
\usepackage[normalem]{ulem}
\usepackage{soul}
\usepackage{authblk} % For author affiliations

\newtheorem{theorem}{Theorem}

\newtheorem{definition}[theorem]{Definition}
\newtheorem{proposition}[theorem]{Proposition}

\newcommand\myeq{\mathrel{\stackrel{\makebox[0pt]{\mbox{\normalfont\tiny def}}}{=}}}

\newcommand{\Cons}{\text{C}}

\DeclareMathOperator{\dis}{ \bf dis}

\definecolor{tocorrect}{rgb}{0.97, 0.04, 0.56}

\SetKwInput{KwInput}{Input}                % Set the Input
\SetKwInput{KwOutput}{Output}              % set the Output

\newcounter{lnote}

\newcommand{\comment}[1]{}

\begin{document}

\title{A continuation method in Bayesian inference}
%
%\date{}
\date{\vspace{-5ex}}
\author[1]{Ben Mansour Dia~\thanks{College of Petroleum Engineering and Geosciences (CPG), King Fahd  University of Petroleum and Minerals (KFUPM), Dhahran 31261, Saudi Arabia, ben.dia@kfupm.edu.sa}}

\maketitle
\begin{abstract}
We present a continuation method that entails generating a sequence of transition probability density functions from the prior to the posterior in the context of Bayesian inference for parameter estimation problems. The characterization of transition distributions, by tempering the likelihood function, results in a homogeneous nonlinear partial integro-differential equation whose existence and uniqueness of solutions are addressed. The posterior probability distribution comes as the interpretation of the final state of the path of transition distributions. A computationally stable scaling domain for the likelihood is explored for the approximation of the expected deviance, where we manage to hold back all the evaluations of the forward predictive model at the prior stage. It follows the computational tractability of the posterior distribution and opens access to the posterior distribution for direct samplings. To get a solution formulation of the expected deviance, we derive a partial differential equation governing the moments generating function of the log-likelihood. We show also that a spectral formulation of the expected deviance can be obtained for low-dimensional problems under certain conditions. The computational efficiency of the proposed method is demonstrated through three differents numerical examples that focus on analyzing the computational bias generated by the method, assessing the continuation method in the Bayesian inference with non-Gaussian noise, and evaluating its ability to invert a multimodal parameter of interest. 
 \\[2pt]

\\[2pt]

% 94A15   	Information theory, general [See also 62B10, 81P45]
% 62F15   	Bayesian inference
% 45E05   	Integral equations with kernels of Cauchy type [See also 35J15]
% 46E15   	Banach spaces of continuous, differentiable or analytic functions
% 34B09   	Boundary eigenvalue problems
% 35F31   	Initial-boundary value problems for nonlinear first-order equations
% 35B60   	Continuation and prolongation of solutions [See also 58A15, 58A17, 58Hxx]
% 65C05   	Monte Carlo methods
\end{abstract}

\section{Introduction}\label{sec1}

There is a broad consensus that in statistical problems for parameter estimation under uncertainty, the Bayesian approach provides the method of choice since it measures the uncertainty in the knowledge of the parameter of interest with the interpretation of the probability density function (PDF). The Bayesian approach to inverse problems leads to the theory of updating a prior PDF  with the information carried by the data, see \cite{DashtiStuart2017, Stuart2010}, through the Bayes' rule, with the aim of computing the expected value of a quantity of interest from the posterior PDF. One important challenge encountered, when deploying the Bayes' theorem, is the non-accessibility of the posterior PDF, which stems from the fact that the normalizing factor, also known as the Bayes factor, is unknown. While significant strides have been made in designing sampling approaches that have access to the posterior PDF, the overwhelming majority of pertinent methods are either computationally demanding in terms of the number of evaluations of the forward predictive model or intractable. 
The computational tractability refers here to the ability of counting, a priori, the number of evaluations of the forward predictive model when approximating the posterior distribution. 
\\[0.25pt]

The most common way of approximating the Bayes factor is to refer to it as the marginal distribution of the likelihood function.  That approach is computationally expensive when it comes to estimate statistics under the posterior PDF. That intensive computational cost is somehow due to the large standard deviation of the prior PDF, rises the curse of dimensionality for deterministic integration methods, and likely suffers from numerical underflow for a vague prior knowledge. To overcome this issue, the importance sampling has been utilized in sophisticated settings: with Laplace approximation for estimating the expected information gain \cite{Beck2018, Beck2019}, based on a multilayered construction to quantify the nuisance effect of some components of the parameter of interest on the expected information gain \cite{Feng2019}, and in an iterative form suitable to being employed in parallel computer architectures \cite{Morzfeld2018}.

Alternatively, advance classes of Markov chain Monte Carlo (MCMC) have been constructed to address the computation of statistics under the posterior PDF, with the basic idea of overpassing the normalizing constant approximation.  Therein, the intensive computational load remains a major issue.  Moreover, the acceptance criterion in MCMC methods deters the ability to count the number of evaluations of the forward model and therefore produces the intractability of the posterior PDF.  A careful and comprehensive analysis based on the non-Markovian aspect of the counting index process is presented in \cite{Syed2019}.
A methodological approach of inspecting the transition probabilities between the prior and the posterior distributions, commonly referred to as continuation methods, has been addressed to promote MCMC methods \cite{Heng2016, Kouetal2006}.
\\[0.25pt]

We present a continuation method to address two challenging issues in Bayesian inverse problems, that of overcoming the computational intractability of the posterior distribution, and scanning the transition distributions between the prior and the posterior. We consider the continuous update of the distribution alongside to the optimal transport approach addressed in a series of papers \cite{ handbook, Moselhy2012, Reich2014}. The continuation method introduced in  \cite{Reich2014}, and used in \cite{Heng2016}, formulates the inversion as optimal control of information transport from the prior to posterior where the PDF is interpreted by the solution of a Liouville type equation with the assumption of available partial observations. 

Our approach is also an alternative that bypasses the normalizing factor estimation via marginal likelihood, and allows direct sampling from the posterior distribution in particular, and from any transition distribution between the prior and the posterior stages in general.  As in \cite{Heng2016}, we characterize the transition PDF when connecting the prior to the posterior by tempering the likelihood so that a continuous analysis of the transition PDF can be performed to assess the information carried by the data. Thereupon, we address also the computational tractability aspect with a continuous and differentiable path sampling to model the update of the prior distribution up to the posterior stage with the only assumption of continuity of the forward model with respect to the parameter of interest.

The starting point consists of composing the likelihood function with a normalized power function, and then we derive the associated transition distribution, named \textit{power posterior distribution}, using the associated Bayes' theorem.  The characterization of the power posterior distribution results in a partial integro-differential equation (PIDE) whose solution at the final state is the posterior PDF. The existence and uniqueness of the solution of the characteristic equation follow from the analysis of boundary value problems of PIDE of Barbashin type presented in \cite{Appell2009, Appell1988, Appell1994, Appell2000,  Kalitvin2013}. 

The computational tractability of the posterior PDF is reduced to the computational tractability of the expected deviance which, in turn, is obtained in two different ways. The first one assumes a null skewness condition and is derived from the spectral expansion of the kernel associated with the log-likelihood. The second formulation is built through the solution of a partial differential equation governing the moments generating function of the log-likelihood function. 

We assess the effectiveness of our method with three numerical examples. The first example is a one-dimensional algebraic forward model with additive Gaussian noise, where analytical computations facilitate an error analysis. The second example is a source inversion problem, where the forward model is the solution of the one-dimensional wave equation. In that problem, we test our method to a complicated setting of the measurement errors formed by the combination of speckle noise and uniform noise. The Wasserstein distance is used to incorporate the noise in the likelihood definition. The objective is to infer the source location and the constant amplitude of waves propagating in a one-dimensional direction. The third example is a multimodal inversion that constitutes a challenging benchmark test for sampling methods. 
\\[0.25pt]

Our contributions are presented as follows. The power posterior notion and its characterization are presented in Section \ref{sec:ppd}. In Section \ref{sec:solv}, we detail the mathematical analysis to address the existence and the uniqueness of the characteristic PIDE equation. Explicit solution formulations for tractable transition distributions are established in Section \ref{sec:sol}. Section \ref{sec:num} is dedicated to numerical experiments performed to evaluate our continuation method.

\section{Power-posterior formulation} \label{sec:ppd}
In this section, we start by stating the Bayesian inversion problem and introduce the continuation method by tempering the likelihood. The characterization of the transition distribution is also presented.

\subsection{Bayesian inversion}
Let $d$ be a non-negative integer, $\bm{\theta}_t \in \mathbb{R}^d$ be the true vector value of an unknown parameter that operates as an input in a forward predictive model denoted by $\bm{g}$. We address the inference of the model parameter $\bm{\theta}_t$, given the forward model $\bm{g}$ and noisy dataset $\bm{Y} \in \mathbb{R}^q$ formed of $q$ observed quantities of the forward model output $\bm{g}(\bm{\theta}_t)$ polluted with measurement noise. We consider computationally intensive Bayesian inverse problems, where almost the entire computational load is associated with the evaluation of $\bm{g}$.

From now,  the parameter of interest is assumed to be a random variable vector, denoted by $\bm{\theta}$ in lieu of $\bm{\theta}_t$, and is characterized by a prior PDF $\pi(\bm{\theta})$ with $\bm{\theta} = (\theta_1(\omega_1), \cdots, \theta_d(\omega_d))^T$, $\omega_i \in \Omega$, where $\Omega$ is the set of all random events. The density $\pi(\bm{\theta})$ is a prior PDF on $\bm{\theta}_t$ that is available to us, and the inference consists in updating that knowlegde on the parameter to the posterior density $\pi(\bm{\theta} | \bm{Y})$ through the data $\bm{Y}$ with distribution $p(\bm{Y})$. For that, we consider the Bayes' theorem given by
\begin{eqnarray}\label{eq:bayes}
  \pi(\bm{\theta} | \bm{Y}) = \frac{p(\bm{Y} | \bm{\theta})\pi(\bm{\theta})}{p(\bm{Y})},
\end{eqnarray}
where $p(\bm{Y} | \bm{\theta})$ is the likelihood, i.e. the function that measures  how the dataset $\bm{Y}$ is distributed for each possible value  $\bm{\theta}$ of the parameter. The selection of the likelihood is crucial in the Bayesian approach, is associated with the type of the measurement errors, and leads to the model specification. Efficient Bayesian inference, in terms of assessing the measurement errors and covering as well the model misspecification, appeals to designing the likelihood function through an optimal mapping problem for the transport of the information \cite{GangboMcCann96, Villani2003}. In this work, we focus on a general exponential-type likelihood function covering complicated noise structures, and given by
\begin{eqnarray}
\label{eq:likelihood}
 p(\bm{Y} | \bm{\theta}) =  \Cons \exp \left( - s \dis(\bm{Y}, \bm{g} (\bm{\theta}))\right).
\end{eqnarray}
Here $s$ is a parameter that quantifies the dispersion of  the measurement error on the data, $\Cons$ is a constant depending on $d$ and $s$, i.e  $\Cons \equiv \Cons(d,s)$, and $\dis$ denotes a given quadratic distance between the data sample and the forward model output sample. Possible choices for the distance $\dis$ are the Wasserstein distance, the Hellinger distance, the total variation distance, the Mahalanobis distance. Weighted versions of the mentioned distances by the distribution of the data noise are also admissible.

\subsection{Power Bayes' rule: a continuation method}
We aim to derive a continuation method that consists of producing a sequence of transition distributions, named \textit{power posterior distributions}, between the prior and the posterior stages, and that is computationally tractable. The concept of power posterior, used throughout this paper, is derived from the power Bayes' rule, and based on the \textit{thermodynamic integration}  that is a widely used for computing the Bayes factor in physical applications, see \cite{FrielHurnWyse2014, FielPettitt2008, Lartillot2006}. It is a kind of a geometric path sampling known as \textit{parallel tempering method} see \cite{GelmanMeng1998, EarlDeem2005} or \textit{replica-exchange method} in molecular dynamics see \cite{SugitaOkamoto1999, Liu2005}, which are devised to advance MCMC methods.

The power Bayes' rule is a readjusted Bayes' rule in the framework of the continuation of distribution from the prior to the posterior. For a given real $\alpha \in [0,1]$, the power posterior distribution is given by
\begin{eqnarray}
\label{eq:powerBayes}
 \pi(\bm{\theta}\vert\bm{Y};\alpha) = \frac{p(\bm{Y}\vert \bm{\theta})^\alpha \pi(\bm{\theta})}{z(\bm{Y}\vert \alpha)},
\end{eqnarray}
where $z(\bm{Y}\vert \alpha)$ is the normalizing constant and can be expressed as the marginal power likelihood. 
The variable $\alpha$ is the inverse tempering temperature, which is named here \textit{tempering parameter} for easy convenience. 
For the tempering parameter $\alpha$ varying from $0$ to $1$, the power posterior PDF encodes the knowledge about the parameter of interest from the prior $\pi(\bm{\theta}) = \pi(\bm{\theta}\vert\bm{Y};\alpha)_{\vert \alpha = 0}$ to the posterior $\pi(\bm{\theta}\vert\bm{Y}) = \pi(\bm{\theta}\vert\bm{Y}; \alpha)_{\vert \alpha =1}$, and the normalizing factor $z(\bm{Y}\vert \alpha)$ carries the data information of the "full evidence" $1 = z(\bm{Y}\vert\alpha)_{\vert \alpha = 0}$ to the sample evidence $p(\bm{Y})= z(\bm{Y}\vert\alpha)_{\vert \alpha = 1}$. 

The notation $\mathbb{E}$ stands for the expectation operator (with respect to the prior PDF), a subscript will be used to specify the associated random variable whenever there is a need for extra clarity. The Leibniz integral rule and the property of interchanging the integral over $\Theta$, the space of the parameter of interest, and the derivation with respect to the tempering parameter $\alpha$ are widely used throughout the manuscript. We will comment on it whenever the utilization seems not obvious.

Next, we characterize the power posterior PDF $\pi(\bm{\theta}|\bm{Y};\alpha)$ with a dynamical equation of which the solution at the final state $\alpha =1$ stands for the posterior PDF.

\subsection{Characterization of the power posterior PDF}
We present the derivation of the equation governing the power posterior PDF $\pi(\bm{\theta}|\bm{Y};\alpha)$. Before all, it is worth mentioning that  the power posterior PDF concentrates when the tempering parameter $\alpha$ grows.
That observation depicts the knowledge updating from the prior to the posterior. To characterize the power posterior, we differentiate the rule \eqref{eq:powerBayes} with respect to the tempering variable $\alpha$, and we get
\begin{eqnarray}
 \frac{\partial \pi(\bm{\theta}|\bm{Y};\alpha)}{\partial \alpha} 
 &=&   \frac{\partial }{\partial \alpha} \left( \frac{ p^\alpha (\bm{Y} | \bm{\theta} )}{z(\bm{Y} | \alpha)} \right) \pi(\bm{\theta})\nonumber \\
 &=& \left( \frac{ \log p(\bm{Y}|\bm{\theta}) p^\alpha (\bm{Y} | \bm{\theta} ) }{z(\bm{Y} | \alpha)} - \frac{ p^\alpha (\bm{Y} | \bm{\theta} ) }{z^2(\bm{Y} | \alpha)} \int_\Theta  \log p(\bm{Y}|\bm{\theta}) p^\alpha (\bm{Y} | \bm{\theta} ) \pi(\bm{\theta}) \hbox{d} \bm{\theta}  \right)\pi(\bm{\theta})\nonumber \\
 &=&   \log p(\bm{Y}|\bm{\theta})\frac{ p^\alpha (\bm{Y} | \bm{\theta} ) \pi(\bm{\theta})}{z(\bm{Y} | \alpha)} - \frac{ p^\alpha (\bm{Y} | \bm{\theta} ) \pi(\bm{\theta}) }{z(\bm{Y} | \alpha)} \mathbb{E}_{\bm{\theta}|\bm{Y};\alpha}\left[ \log p(\bm{Y}|\bm{\theta})\right]\nonumber \\
 &=& \Big( \log p(\bm{Y}|\bm{\theta}) - \mathbb{E}_{\bm{\theta}|\bm{Y};\alpha} \left[\log p(\bm{Y}|\bm{\theta})\right] \Big) \pi (\bm{\theta} | \bm{Y}; \alpha). \nonumber
\end{eqnarray}
The expectation of the log-likelihood with respect to the power posterior PDF, $ \mathbb{E}_{\bm{\theta}|\bm{Y};\alpha} \left[\log p(\bm{Y}|\bm{\theta})\right]$, is known as the \textit{expected deviance}. It is used in Bayesian model selection to assess the fitting goodness as it generalizes the Akaike information criterion (AIC), see \cite{Francois2011}.  We are now in the position to state the characterization result.
\begin{proposition}
The power posterior distribution solves the partial integro-differential equation
\begin{eqnarray}
\label{eq:pide}
\frac{\partial \pi(\bm{\theta}|\bm{Y};\alpha)}{\partial \alpha}  = \Big( \log p(\bm{Y}|\bm{\theta}) - \mathbb{E}_{\bm{\theta}|\bm{Y};\alpha} \left[\log p(\bm{Y}|\bm{\theta})\right] \Big) \pi (\bm{\theta} | \bm{Y};\alpha),\quad \hbox{for}\quad 0< \alpha  \leq 1,
\end{eqnarray}
with initial condition 
\begin{eqnarray*}
\pi(\bm{\theta}|\bm{Y};\alpha = 0) = \pi(\bm{\theta}).
\end{eqnarray*}
\end{proposition}
Similar characterization of transition distributions with a general tempering function is presented in \cite{Heng2016}, where the expected deviance is seen as a reference value that controls the evolution of $\pi(\bm{\theta}|\bm{Y};\alpha)$ with respect to the tempering parameter $\alpha$. Our interests in the PIDE \eqref{eq:pide} is two-fold. 
\begin{itemize}
\item The first issue is subsequent to its derivation and  is the solvability. By this we mean the existence and the uniqueness of a solution. Note that the characteristic equation has a type of logistic dynamics with intrinsic decay rate of 
\begin{eqnarray*}
 r = \frac{\log p(\bm{Y}|\bm{\theta}) - \mathbb{E}_{\bm{\theta}|\bm{Y};\alpha} \left[\log p(\bm{Y}|\bm{\theta})\right]}{s} =  \underbrace{\mathbb{E}_{\bm{\theta}|\bm{Y};\alpha} \left[\dis(\bm{Y}, \bm{g}(\bm{\theta}))\right]}_{(*)}  - \underbrace{\dis(\bm{Y}, \bm{g}(\bm{\theta}))}_{(**)}
\end{eqnarray*}
 of which the sign is not known in advance; each $\bm{\theta}$ defines a direction to follow.
Therefore resolution approaches for dynamical systems can be applied to \eqref{eq:pide} with the constraints of sampling $\bm{\theta}$ in $(*)$ from the power posterior $\pi(\bm{\theta}|\bm{Y};\alpha)$ while $\bm{\theta}$ in $(**)$ is drawn from the prior distribution $\pi(\bm{\theta})$. We adopt a different approach for the solvability in Section \ref{sec:solv}, where  \eqref{eq:pide} is recast in a form that is suitable to address its solvability.
\item Second, we aim to devise computationally tractable solutions of \eqref{eq:pide} in Section \ref{sec:sol}. We address the tractability of the expected deviance that subsequently implies the tractability of the power posterior distributions.  A spectral approach for low dimensional problems with null skewness condition gives an unbiased formulation of the power posterior.
\end{itemize}

\section{Solvability of the characteristic equation} \label{sec:solv} 
This section addresses the existence and the uniqueness of the solution of the transition probability. We start by extending the characteristic equation \eqref{eq:pide} from the diagonal subspace to the entire space $\Theta \times \Theta$. %using the assumption the pseudo-euclidean space for the data space $\mathbb{R}^q$. 
The transition probability distribution is then recast as co-factors product for sake of clarity.

\subsection{Reformulation of the characteristic equation}
For sake of conciseness and as there is no confusion, the representation of the data $\bm{Y}$ (it is not a variable in \eqref{eq:pide}) is omitted until Section \ref{sec:sol}. Hence we introduce
\begin{eqnarray*}
\tilde{\pi}(\alpha, \bm{\theta}) &\myeq & \pi(\bm{\theta}|\bm{Y};\alpha),\\
C_0 & \myeq & \dis(\bm{Y}, \bm{g}(\bm{\theta})) \quad \hbox{with} \quad  \bm{\theta}\quad \hbox{drawn from} \quad \tilde{\pi}(\alpha=0, \bm{\theta}) = \pi(\bm{\theta}),\\
\mathcal{K}(\alpha, \bm{\eta},\bm{\eta}) & \myeq &\dis(\bm{Y}, \bm{g}(\bm{\eta}))  \quad \hbox{with} \quad \bm{\eta} \quad \hbox{drawn from} \quad \tilde{\pi}(\alpha, \bm{\theta}).
\end{eqnarray*}
The prior $\pi(\bm{\theta})$ is the initial condition for the dynamics of the power posterior PDF with the tempering variable $\alpha$, therefore the functional $C_0$ reads only the initial state, and does not depend on the power posterior PDF $\tilde{\pi}(\alpha, \bm{\theta})$ for $0<\alpha \leq 1$. By definition, $\mathcal{K}$ is a symmetric kernel, and the expected deviance is recast as
\begin{eqnarray}
\mathbb{E}_{\bm{\theta}|\bm{Y};\alpha} \left[\log p(\bm{Y}|\bm{\theta})\right] = \log(\Cons) -s\int_\Theta \mathcal{K}(\alpha, \bm{\theta},\bm{\theta}) \tilde{\pi}(\alpha, \bm{\theta}) \hbox{d} \bm{\theta},
\end{eqnarray}
therefore the PIDE \eqref{eq:pide}, governing the power posterior distribution $\tilde{\pi}(\alpha, \bm{\theta})$, becomes 
\begin{eqnarray}
\label{eq:pide2}
\left\{
\begin{array}{ll}
 \displaystyle{\frac{\partial \tilde{\pi}(\alpha, \bm{\theta})}{\partial \alpha} = - sC_0 \tilde{\pi}(\alpha, \bm{\theta}) + s\tilde{\pi}(\alpha, \bm{\theta}) \int_\Theta \mathcal{K}(\alpha, \bm{\eta},\bm{\eta}) \tilde{\pi}(\alpha, \bm{\eta}) \hbox{d} \bm{\eta},}
 \\\\
 \displaystyle{\tilde{\pi}(0, \bm{\theta}) = \pi(\bm{\theta}).}
 \end{array}
 \right.
\end{eqnarray}
The above equation is a nonlinear PIDE with nonlinear kernel, and addressing its solvability requires efforts on tackling an extra challenge due to the fact that $\mathcal{K}(\alpha, \cdot,\cdot)$ maps from only the diagonal sub-space of the space $\Theta \times \Theta$.

We address the solvability by extrapolating equation \eqref{eq:pide2} over the whole space $\Theta \times \Theta$. Typically, in lieu of \eqref{eq:pide2}, we consider the following equation 
 \begin{eqnarray}
 \label{eq:pide3}
\left\{
\begin{array}{ll}
 \displaystyle{\frac{\partial \tilde{\pi}(\alpha, \bm{\theta})}{\partial \alpha} = - sC_0 \tilde{\pi}(\alpha, \bm{\theta}) + s\tilde{\pi}(\alpha, \bm{\theta}) \int_\Theta \mathcal{K}(\alpha, \bm{\theta},\bm{\eta}) \tilde{\pi}(\alpha, \bm{\eta}) \hbox{d} \bm{\eta},}
 \\\\
 \displaystyle{\tilde{\pi}(0, \bm{\theta}) = \pi(\bm{\theta}).}
 \end{array}
 \right.
\end{eqnarray}
 To give sense to the kernel $ \mathcal{K}(\alpha, \bm{\theta},\bm{\eta})$ over the space $\Theta \times \Theta$, we assume that the data space, that is a subspace of $\mathbb{R}^q$, is a pseudo-euclidean space, i.e. the projection over each single component of a dataset forms an euclidean line. That leads to the interpretation of
 \begin{eqnarray*}
 \mathcal{K}(\alpha, \bm{\theta},\bm{\eta}) = \langle \bm{Y} - \bm{g}(\bm{\theta}), \bm{Y} - \bm{g}(\bm{\eta}) \rangle,
 \end{eqnarray*}
 where the inner product  $\langle \cdot, \cdot \rangle$ is associated to the quadratic distance $\dis$.

System \eqref{eq:pide2} is a special case of \eqref{eq:pide3} over the diagonal subspace of $\Theta \times \Theta$, then the solvability of \eqref{eq:pide3} implies  the solvability of \eqref{eq:pide2}. Indeed, equation \eqref{eq:pide2} is in the form $ \frac{\partial \tilde{\pi}}{\partial \alpha}(\alpha, \bm{\theta}) = r \tilde{\pi}(\alpha, \bm{\theta})$ while \eqref{eq:pide3}  has the form $ \frac{\partial \tilde{\pi}}{\partial \alpha}(\alpha, \bm{\theta}) = r(\bm{\theta}) \tilde{\pi}(\alpha, \bm{\theta})$, where $\bm{\theta}$ is sampled from $\tilde{\pi}(\alpha, \bm{\theta})$ for $\alpha >0$. 
%Next, we present an approach to study solvability over the entire space.
%
Looking at the functional $\tilde{\pi}(\alpha, \bm{\theta})$ in the form
\begin{eqnarray*}
  \tilde{\pi}(\alpha, \bm{\theta}) = u(\alpha, \bm{\theta}) \exp ( v(\alpha, \bm{\theta})),
\end{eqnarray*}
where for $(\alpha, \bm{\theta}) \in ]0,1] \times \Theta$, it follows that  $u$ and $v$ solve 
\begin{eqnarray}
\label{eq:dif1}
 (a)\left\{
\begin{array}{ll}
\displaystyle{\frac{\partial u (\alpha, \bm{\theta})}{\partial \alpha} = - C_0 u (\alpha, \bm{\theta}),}&  \\\\
\displaystyle{ u(0, \bm{\theta}) = \pi(\bm{\theta}),}
\end{array}
\right.
\quad \quad
(b) \left\{
\begin{array}{ll}
\displaystyle{\frac{\partial v (\alpha, \bm{\theta})}{\partial \alpha} = \int_\Theta \mathcal{H}(\alpha, \bm{\theta}, \bm{\eta}, v(\alpha,\bm{\eta}))  \hbox{d}\bm{\eta},}\\\\
\displaystyle{ v(0, \bm{\theta}) = 0, }
\end{array}
\right.
\end{eqnarray}
respectively, with 
$$\mathcal{H}(\alpha, \bm{\theta}, \bm{\eta}, v(\alpha,\bm{\eta})) = u(\alpha,\bm{\theta}) \mathcal{K}(\alpha, \bm{\theta}, \bm{\eta}) \exp( v(\alpha,\bm{\eta})).$$ 
We have reduced the solvability of \eqref{eq:pide3} to the solvability of (\ref{eq:dif1}b) because (\ref{eq:dif1}a) admits a solution of form $u (\alpha, \bm{\theta}) = \pi(\bm{\theta}) \exp (-\alpha s C_0)$. Equation (\ref{eq:dif1}b) is a particular case of PIDE of Barbashin type \cite{Appell2000}, and is a continuous analogue to countable systems of ordinary differential equations of the form
\begin{eqnarray*}
\frac{d \tilde{v}_i (\alpha)}{ d \alpha} = \sum_{j=1}^{\infty} h_{jj}(\alpha) \tilde{v}_j(\alpha)\tilde{v}_i(\alpha)\quad i=1,2,\cdots,
\end{eqnarray*}
where $\frac{d}{d \alpha}$ is in the Fr\'echet sense. An useful observation from the above equation is that one can deal with the solvability of  (\ref{eq:dif1}b) by referring to it as a class of abstract differential equations in the appropriate Banach space. To do so, we next delineate the framework by setting some basic notions.

\subsection{Basic statements}
To define the space of probability density functions, we assume that $\Theta$ is a product space of $d$ bounded intervals $\Theta = \prod_{i=1}^{d}[a_i,b_i]$, where $a_i,b_i \in \mathbb{R}$. Extension of the marginal support $[a_i, b_i]$ to $\mathbb{R}$ (e.g. for Gaussian density function) is straightforward provided that the density function vanishes at the limits. The space of probability density functions is given by
\begin{eqnarray*}
 \displaystyle{\mathcal{S} = \left\{\mu \in \mathcal{C} \left( \Theta, \mathbb{R}^+\right)\; : \quad \int_\Theta \mu(\bm{\theta}) \hbox{d} \bm{\theta} = 1 \right\},}
\end{eqnarray*}
where $\mathcal{C} \left(\Theta, \mathbb{R}^+\right)$ denotes %positve real-valued function that is differentiable with respect to the first argument, and 
the set of all real positive continuous functions in $\Theta$, and the condition $\int_\Theta \mu(\bm{\theta}) d \bm{\theta} = 1$ implies that $\mu$ is bounded. 
\begin{definition}[Composition laws]
For $r \in \mathbb{R}^+$ and $\mu_1, \mu_2 \in \mathcal{S}$, the inner operator $\oplus : \mathcal{S} \times \mathcal{S} \rightarrow \mathcal{S}$ and the external operator $\odot: \mathbb{R}^+\times \mathcal{S}\rightarrow \mathcal{S}$ given by
\begin{eqnarray*}
\displaystyle{\mu_1 \oplus \mu_2 = \frac{\mu_1( \bm{\theta})\mu_2( \bm{\theta})}{\int_\Theta \mu_1(\bm{\theta})\mu_2( \bm{\theta}) \hbox{d}\bm{\theta}},\quad\quad \hbox{and}\quad \quad  
r \odot \mu_1 = \frac{\mu_1^r( \bm{\theta})}{\int_\Theta \mu_1^r(\bm{\theta}) \hbox{d}\bm{\theta}}}
\end{eqnarray*}
define respectively algebraic composition laws for the perturbation and power transformations in $\mathcal{S}$.
\end{definition}
\begin{proposition}\label{propo3}
The space $\mathcal{S}$ with the perturbation law $\oplus$, the power law $\odot$, and the norm $\|\cdot\|$ induced by the distance $\dis$ is a Banach space of real functions in $\Theta$.
\end{proposition}
\begin{proof}
Extending the one-dimensional analysis developed in \cite{Egozcue2006} to our problem at hand (setup and $d$-dimensional space) shows that $\left(\mathcal{S}, \oplus, \odot\right)$ is a vector space. It is well known that every finite-dimensional normed vector space is a complete, $\mathcal{S}$ equipped with the norm induced by the  distance $\dis$ is then a Banach space. % as $\mathcal{S}$ is of dimension $d+1$.
\end{proof}
As corollary from Proposition \ref{propo3}, the sub-space $\mathcal{S}_{l}$ of $\mathcal{S}$ given by
\begin{eqnarray*}
  \displaystyle{\mathcal{S}_l = \left\{\mu \in \mathcal{S}  : \quad |\log \mu | < \infty \right\}}
\end{eqnarray*}
is a Banach space, where the completness follows from the concavity of the logarithmic function.

\subsection{Solvability of the Barbashin equation}
It is shown in  \cite{Appell2009, Appell1994, Appell2000, Kalitvin2013} that the solvability of Barbashin equation is given by the strong continuity of the associated abstract integral operator. We now state the solvability result.
\begin{theorem}
\label{theorem2}
Suppose that the forward operator $\bm{g}$ is continuous with respect to the parameter of interest $\bm{\theta}$, then the partial integro-differential equation of Barbashin type (\ref{eq:dif1}b) has a unique solution $v$ that is bounded in the Banach space $\mathcal{S}_l$.
\end{theorem}
The basis of the proof, presented in appendix \ref{sec:app1}, is to write out the associated Cauchy operator and to show its strong continuity. For any $\alpha$, we have shown the existence and the uniqueness of the power posterior distribution $\tilde{\pi}(\alpha, \bm{\theta})$ . In the next section, we address the solution formulation.

\section{Solution formulations} \label{sec:sol}
In this section, we re-highlight the dependence on the data $\bm{Y}$, and refer back to the notation $\pi(\bm{\theta}|\bm{Y};\alpha)$ for the power posterior distribution. The PIDE \eqref{eq:pide} has the solution of form
\begin{eqnarray}
\label{eq:solppd}
   \pi(\bm{\theta}|\bm{Y};\alpha) = \pi(\bm{\theta})  p^\alpha (\bm{Y}|\bm{\theta}) \exp\left(-\int_0^\alpha \mathbb{E}_{\bm{\theta}|\bm{Y}; \tau}\left[\log p(\bm{Y}|\bm{\theta}) \right] \hbox{d} \tau \right)\quad \hbox{for}\quad \alpha \in [0,1].
\end{eqnarray}
The power posterior PDF $\pi(\bm{\theta}|\bm{Y};\alpha)$ given at \eqref{eq:solppd} is computationally intractable  as the evaluation of the expected deviance $ \mathbb{E}_{\bm{\theta}|\bm{Y}; \alpha}\left[\log p(\bm{Y}|\bm{\theta}) \right] $ requires the use of the distribution $\pi(\bm{\theta}|\bm{Y};\alpha)$. Another observation from \eqref{eq:solppd} is that the normalizing Bayes factor $z(\bm{Y}|\alpha)$, depends fully on the sum of the expected deviance from the prior stage up to the stage $\alpha$. Our goal in this section is to come up with a setting of $ \mathbb{E}_{\bm{\theta}|\bm{Y}; \alpha}\left[\log p(\bm{Y}|\bm{\theta}) \right] $ that systematically holds back those dependencies to the prior stage. We will do so by devising sampling process, without performing MCMC or using importance sampling, from the prior  only to evaluate $\mathbb{E}_{\bm{\theta}|\bm{Y}; \alpha}\left[\log p(\bm{Y}|\bm{\theta}) \right]$. To tract efficiently the transition distributions along the tempering parameter $\alpha$, we address the computational tractability of the  expected deviance. We propose two approaches, one using the spectral expansion of the kernel $\mathcal{K}$, and another based on the moments generating function of the log-likelihood function.

\subsection{Spectral formulation for low dimensional problems}
\label{sec:spec}
In low-dimensional problems, many numerical approximation approaches, among those the spectral methods, work efficiently and even outperform the Monte Carlo method in some scenarios. In this first formulation, we consider a continuous learning, from the prior stage to any transition stage $\alpha$, $0 < \alpha \leq 1$, of low-dimensional parameter of interest $\bm{\theta}$, and use the symmetric structure of the log-likelihood function (null skewness condition), to approximate the expected deviance as the sum of the eigenvalues of the kernel $\mathcal{K}$. The functional $\mathcal{K}$ is a positive symmetric kernel, thanks to the Mercer theorem \cite{Santin2016}, $\mathcal{K}$ can be written as
\begin{eqnarray}
  \mathcal{K}(\alpha, \bm{\theta}, \bm{\eta}) = \sum_{n=1}^{\infty} \lambda_n(\alpha)k_n^\alpha(\bm{\theta})k_n^\alpha(\bm{\eta}),
\end{eqnarray}
where the eigenvalues $\lambda_n$ are positive, and the eigenfunctions $\{k_n^\alpha\}_{n>0}$ form an orthonormal basis with respect to the power posterior $\pi(\bm{\theta}|\bm{Y};\alpha)$,
\begin{eqnarray}
\label{eq:orth}
\left\langle k^\alpha_n,k^\alpha_m\right\rangle_{\alpha} \myeq \int_\Theta k^\alpha_n(\bm{\theta}) k^\alpha_m(\bm{\theta}) \pi(\bm{\theta}|\bm{Y};\alpha) \hbox{d} \bm{\theta} = \left\{
\begin{array}{ll}
1 \quad \hbox{if}\quad n=m,\\\\
0\quad  \hbox{otherwise}.
\end{array}
\right.
\end{eqnarray} 
Therefore, the expected deviance is seen as the trace of the kernel denoted by $L_1(\alpha)$, and given by
\begin{eqnarray*}
L_1(\alpha) \myeq  \int_\Theta  \mathcal{K}(\alpha, \bm{\eta}, \bm{\eta}) \pi(\bm{\theta}|\bm{Y};\alpha) \hbox{d}\bm{\eta} 
 &=&   \sum_{n=1}^{\infty} \int_\Theta   \lambda_n(\alpha)k^\alpha_n(\bm{\eta})k^\alpha_n(\bm{\eta}) \pi(\bm{\theta}|\bm{Y};\alpha) \hbox{d}\bm{\eta}\\
 &=&  \sum_{n=1}^{\infty} \lambda_n(\alpha) \int_\Theta  k^\alpha_n(\bm{\eta})k^\alpha_n(\bm{\eta}) \pi(\bm{\theta}|\bm{Y};\alpha) \hbox{d}\bm{\eta}\\
 &=&  \sum_{n=1}^{\infty} \lambda_n(\alpha).
\end{eqnarray*} 
The eigenpair $(\lambda_n, k^\alpha_n)$ solves the following eigenvalue problem
\begin{eqnarray} 
\label{eq:eigen}
 \int_\Theta  \mathcal{K}(\alpha, \bm{\theta}, \bm{\eta})k^\alpha_n(\bm{\eta}) \pi(\bm{\eta}|\bm{Y};\alpha) \hbox{d} \bm{\eta} = \lambda_n(\alpha) k^\alpha_n(\bm{\theta}) \quad \hbox{for}\quad n\geq 1.
\end{eqnarray}
Likewise, the trace $L_2$ of the symmetric positive kernel $\mathcal{K}^2$, the Hadamard product of $\mathcal{K}$ by itself, is given by $L_2(\alpha) = \sum_{n=1}^{\infty} \lambda_n^2(\alpha)$.  
The traces $L_1$ and $L_2$ are solutions of
\begin{eqnarray}
\label{eq:kerneltrace}
%\left\{
\begin{array}{lll}
\displaystyle{\frac{\partial  L_1(\alpha)}{\partial \alpha}   = 2s\Big( L_1^2(\alpha) - L_2(\alpha) \Big),} \quad \hbox{and} \quad
\displaystyle{ \frac{\partial  L_2(\alpha)}{\partial \alpha}   = 2sL_1(\alpha)\Big( L_1^2(\alpha) - L_2(\alpha) \Big),}&  
\end{array}
%\right.
\end{eqnarray}
with initial conditions $L_1(0) = \sum_{n=1}^{\infty} \lambda_n(0)$ and $L_2(0) = \sum_{n=1}^{\infty} \lambda^2_n(0)$, respectively. The derivation of the kernel trace equations \eqref{eq:kerneltrace} is presented in appendix \ref{sec:app3} and is based on the null skewness condition for the log-likelihood function
\begin{eqnarray}
\label{eq:skewcondition}
 \mathbb{E}_{\bm{\theta}|\bm{Y}; \alpha}\left[\Big(\log p(\bm{Y}|\bm{\theta})  - \mathbb{E}_{\bm{\theta}|\bm{Y}; \alpha}\left[\log p(\bm{Y}|\bm{\theta}) \right]\Big)^3 \right] = 0.
\end{eqnarray}
The condition \eqref{eq:skewcondition} results from the exponential structure of the likelihood function with quadratic distance combined with the norm equivalence in finite dimensional space. It is obvious to see when the forward model $\bm{g}$ is linear in the parameter $\bm{\theta}$ or the quadratic distance $\dis$ is associated to a $L^2$ type norm.

From \eqref{eq:kerneltrace}, it comes that $L_2(\alpha) = \frac{1}{2}L_1(\alpha) + constant$. 
Consequently the power posterior distribution has the form
\begin{eqnarray}
 \pi(\bm{\theta}|\bm{Y};\alpha) = \pi(\bm{\theta}) p^\alpha (\bm{Y}|\bm{\theta}) \exp\left( -\alpha \log(\Cons) + \int_0^\alpha \tan(C_1 \tau + C_2) \hbox{d} \tau\right),
\end{eqnarray}
where $C_1$ and $C_2$ are constants, depending on the data dispersion factor $s$, to be estimated at the prior stage. Evaluations of the forward model are needed only at the initial (prior) state to compute the eigenvalues $\lambda_n(0)$ for the approximation of $L_1(0)$, $L_2(0)$, $C_1$ and $C_2$.

\subsection{Solution formulation via moments generating function (MGF)}
\label{sec:solmgf}
The MGF is an useful computational tool when it comes to estimating the moments of a random variable with given distribution. Although the distribution of $\log p(\bm{Y}|\bm{\theta}(\omega))$ is not available to us,  we address the direct approximation of the expected deviance $\mathbb{E}_{\bm{\theta}|\bm{Y}, \alpha}\left[\log p(\bm{Y}|\bm{\theta}) \right]$ using MGF. Let $\Phi_n$ be the $n^{th}$ moment of the log-likelihood with respect to the power posterior distribution,
\begin{eqnarray}
 \Phi_n(\alpha) = \mathbb{E}_{\bm{\theta}|\bm{Y}; \alpha}\left[\log^n p(\bm{Y}|\bm{\theta}) \right].
\end{eqnarray}
By definition of the likelihood function \eqref{eq:likelihood}, $\Phi_n$ exists and has $\mathcal{C}^{\infty}$ regularity as the property, of interchanging the order of the differentiation, with respect to $\alpha$, with the integration with respect to $\bm{\theta}$, holds, see Lemma 3.1 \cite{Heng2016}. Differentiating $\Phi_n(\alpha)$ with respect to $\alpha$, it follows that 
\begin{eqnarray}
\label{eq:moments}
 \frac{\partial \Phi_n(\alpha)}{\partial \alpha} = \Phi_{n+1}(\alpha) - \Phi_1(\alpha) \Phi_n(\alpha),
\end{eqnarray}
with a given initial state $\Phi_n(0) = \mathbb{E}\left[\log^n p(\bm{Y}|\bm{\theta}) \right] $. The expected deviance $\mathbb{E}_{\bm{\theta}|\bm{Y}; \alpha}\left[\log p(\bm{Y}|\bm{\theta}) \right]$ corresponds to  $\Phi_1$, solution of a system of nested differential equations, where the $(n+1)^{th}$ moment operates as an input in the governing equation of the $n^{th}$ moment. For the spectral setting at Section \ref{sec:spec}, the nested system is closed at rank $n=3$ with the null skewness condition \eqref{eq:skewcondition}, and the resulting curve of the expected deviance maps the prior and the posterior stages. Here, we count for all the moments, and provide a stepwise presentation for sake of clarity.

\subsubsection{Governing equation of the MGF}
In this first step, we look at the log-likelihood as a random variable with randomness inherited from $\bm{\theta}$ only (as we assume that $\bm{Y}$ is given), then the $n^{th}$ order of its expected value satisfies
\begin{eqnarray}
 \frac{\partial \Phi_n(\alpha)}{\partial \alpha} = \Phi_{n+1}(\alpha) - \Phi_1(\alpha) \Phi_n(\alpha),
\end{eqnarray}
for $n\geq 0$, and $\alpha \in ]0,1]$, with $\Phi_0(\alpha) =1$.  For a positive real $T_\beta$ large enough ($T_\beta \approx 1$), we have
\begin{eqnarray*}
 \frac{\beta^n}{n!}\frac{\partial \Phi_n(\alpha)}{\partial \alpha} = \frac{\beta^n}{n!} \Phi_{n+1}(\alpha) - \Phi_1(\alpha) \frac{\beta^n}{n!} \Phi_n(\alpha) \quad \forall n \geq 0 \quad \hbox{and} \quad \beta \in [0,T_\beta].
\end{eqnarray*}
We sum over $n$ up to infinity to obtain
\begin{eqnarray*}
 \sum_{n=0}^{\infty}\frac{\beta^n}{n!}\frac{\partial \Phi_n(\alpha)}{\partial \alpha} = \sum_{n=0}^{\infty}\frac{\beta^n}{n!} \Phi_{n+1}(\alpha) - \Phi_1(\alpha) \sum_{n=0}^{\infty}\frac{\beta^n}{n!} \Phi_n(\alpha).
\end{eqnarray*}
Let $m$ be the MGF of the log-likelihood, it is given by
\begin{eqnarray*}
m(\alpha,\beta) =  \sum_{n=0}^{\infty}\frac{\beta^n}{n!}\Phi_n(\alpha),
\end{eqnarray*}
from which it follows that
\begin{eqnarray*}
 \frac{\partial m(\alpha, \beta)}{\partial \alpha}  
 &=&  \frac{\partial }{\partial \beta} \left( \sum_{n=0}^{\infty}\frac{\beta^{n+1}}{(n+1)!} \Phi_{n+1}(\alpha) \right) - \Phi_1(\alpha) m(\alpha, \beta)\\
& =&  \frac{\partial }{\partial \beta} \left( \sum_{n=0}^{\infty}\frac{\beta^n}{n!} \Phi_{n}(\alpha) - \frac{\beta^0}{0!} \Phi_{0}(\alpha) \right) - \Phi_1(\alpha) m(\alpha, \beta)\\
& = &   \frac{\partial }{\partial \beta} \left( \sum_{n=0}^{\infty}\frac{\beta^n}{n!} \Phi_{n}(\alpha) \right)   - \Phi_1(\alpha) m(\alpha, \beta).
\end{eqnarray*}
The governing equation of the MGF is the following first order partial differential equation (PDE)
\begin{eqnarray}
\label{eq:mgf1}
\frac{\partial m(\alpha, \beta)}{\partial \alpha} = \frac{\partial m(\alpha, \beta)}{\partial \beta} - \Phi_1(\alpha) m(\alpha,\beta), \quad \hbox{for} \quad (\alpha, \beta) \in ]0,1] \times ]0,T_\beta].
\end{eqnarray}
Equation \eqref{eq:mgf1} is a nonlinear advection-reaction equation that has to be equipped with a suitable boundary condition either at $\alpha=0$, $\alpha =1$, or $\beta =0$ to form a wellposed initial-value problem (IVP). Subsequently, applying the method of characteristics suffices to come up with the unique solution $m$.

\subsubsection{Admissible initial conditions}
\label{sec:adic}
In this step, we enumerate all the initial conditions that can be assigned to \eqref{eq:mgf1} and aligned with our objective of not-sampling from any transition distribution, even less from the posterior one. Typically, we interpret the properties of MGF to arrive at the following conditions:
\begin{itemize}
\item{Initial value of MGF at $\beta=0$:}
\begin{eqnarray*}
 m(\alpha,0) = \Phi_0(\alpha) = 1.
 \end{eqnarray*}
\item{Initial value of MGF at the prior stage:}
\begin{eqnarray*}
 m(0, \beta) &=& \sum_{n=0}^{\infty}\frac{\beta^n}{n!}\Phi_n(0)
  =  \sum_{n=0}^{\infty}\frac{\beta^n}{n!}\mathbb{E}_{\bm{\theta}|\bm{Y}; \alpha=0}\left[\log^n p(\bm{Y}|\bm{\theta}) \right]
      = \sum_{n=0}^{\infty}\frac{\beta^n}{n!}\mathbb{E}\left[\log^n p(\bm{Y}|\bm{\theta}) \right]\\
      &=& \sum_{n=0}^{\infty}\mathbb{E}\left[\frac{\beta^n}{n!}\log^n p(\bm{Y}|\bm{\theta}) \right]
      = \sum_{n=0}^{\infty}\mathbb{E}\left[\frac{\left(\beta\log p(\bm{Y}|\bm{\theta})\right)^n}{n!} \right].
\end{eqnarray*}
Let $X_n(\omega) =\frac{\left(\beta\log p(\bm{Y}|\bm{\theta}(\omega))\right)^n}{ n! }$ with $\omega \in \Omega$, as we consider a single observation setting for the likelihood, the inequality $|X_n| \leq \frac{\beta^n}{n!}$ holds almost everywhere, for all $n$. It is obvious that the series of general terms $\frac{\beta^n}{n !}$ is a Taylor series, and obviously convergent. Consequently, we can interchange the integration $\mathbb{E}$ and the summation over non-negative integer $n$ of the mesurable function $X_n$, that is
\begin{eqnarray*}
    m(0, \beta)  = \mathbb{E}\left[ \sum_{n=0}^{\infty} \frac{\left(\beta \log p(\bm{Y}|\bm{\theta})\right)^n}{n!} \right] = \mathbb{E}\left[ \exp \left(\beta \log p(\bm{Y}|\bm{\theta})\right) \right]=\mathbb{E}\left[ p^\beta(\bm{Y}|\bm{\theta}) \right].
\end{eqnarray*}
\end{itemize}
Looking for the initial condition at $\alpha=1$ results in $m(1,\beta)=\mathbb{E}_{\bm{\theta}|\bm{Y}; \alpha=1}\left[ p^\beta(\bm{Y}|\bm{\theta}) \right]$ and requires sampling from the posterior PDF, which we aim to avoid in this work.  Roughly speaking, we shall deal with two candidates $m(\alpha,0)=1$ and $m(0,\beta)= \mathbb{E}\left[ p^\beta(\bm{Y}|\bm{\theta}) \right]$ as initial condition to complete the MGF equation.

\subsubsection{Solving the MGF equation}
We now investigate the analytical resolution of the MGF equation \eqref{eq:mgf1} in the domain $]0,1] \times ]0,+\infty[$
with initial condition $m(\alpha,0) = 1$ or $m(0,\beta) = \mathbb{E}\left[ p^\beta(\bm{Y}|\bm{\theta}) \right]$.  We consider the function $F$, for the change of variable, given by
\begin{eqnarray*}
 F(\alpha, \beta) = m(\alpha, \beta)\exp \left(\int_0^\alpha \Phi_1(\tau) \hbox{d} \tau \right).
\end{eqnarray*}
Differentiating with respect to $\alpha$ and to $\beta$ (by applying the Leibniz integral rule) gives respectively
\begin{eqnarray*}
\frac{\partial F}{\partial \alpha}  &=& \frac{\partial m}{\partial \alpha} \exp \left(\int_0^\alpha \Phi_1(\tau) \hbox{d} \tau \right) + \left(\frac{\partial }{\partial \alpha} \int_0^\alpha \Phi_1(\tau) \hbox{d} \tau \right) m(\alpha, \beta)  \exp \left(\int_0^\alpha \Phi_1(\tau) \hbox{d} \tau \right)\\
&=& \frac{\partial m}{\partial \alpha} \exp \left(\int_0^\alpha \Phi_1(\tau) \hbox{d} \tau \right) +  \Phi_1(\alpha)  m(t, \alpha)  \exp \left(\int_0^\alpha \Phi_1(\tau) \hbox{d} \tau \right)\\
&=& \left(\frac{\partial m}{\partial \alpha}  +  \Phi_1(\alpha)  m(\alpha, \beta) \right) \exp \left(\int_0^\alpha \Phi_1(\tau) \hbox{d} \tau \right),
\end{eqnarray*}
and
\begin{eqnarray*}
\frac{\partial F}{\partial \beta} = \frac{\partial m}{\partial \beta} \exp \left(\int_0^\alpha \Phi_1(\tau) \hbox{d} \tau \right).
\end{eqnarray*}
Plugging the MGF equation \eqref{eq:mgf1}, it is seen that $F$ statifies the simpler equation
\begin{eqnarray}
\label{eq:change_mgf}
 \frac{\partial F}{\partial \alpha} - \frac{\partial F}{\partial \beta} = 0,
\end{eqnarray}
that has to be equipped with initial condition associated with one of the admissible initial conditions of \eqref{eq:mgf1} discussed in \ref{sec:adic}. 

We examine the initial condition to mark out a well-posed IVP governing $F$. The differential equation \eqref{eq:change_mgf} with initial condition  $F_0(\alpha)$, given by
\begin{eqnarray*}
F_0(\alpha) &=& F(0,\alpha) \\
            &=& m(\alpha, 0)\exp \left(\int_0^\alpha \Phi_1(\tau) \hbox{d} \tau \right) \\
            &=& \exp \left(\int_0^\alpha \Phi_1(\tau) \hbox{d} \tau \right),
\end{eqnarray*}
gives an ill-posed IVP in the Hadamard sense as it admits an infinity of solutions. However, the differential equation \eqref{eq:change_mgf} with initial condition $F^0(\beta)$, given by
\begin{eqnarray*}
F^0(\beta) &=& m(0, \beta)
%       &=& m^0(\beta), 
\end{eqnarray*}
has an unique solution of the form $F(\alpha, \beta) = m(0,\beta+\alpha)$. From now on, we drop the initial condition $m(\alpha, 0) =1$ and the MGF equation refers to \eqref{eq:mgf1} and the initial condition $m(0,\beta) = \mathbb{E}\left[ p^\beta(\bm{Y}|\bm{\theta}) \right]$, which together have a unique solution given by
\begin{eqnarray}
 m(\alpha, \beta) = \mathbb{E}\left[ p^{\alpha+\beta}(\bm{Y}|\bm{\theta}) \right] \exp \left(-\int_0^\alpha \Phi_1(\tau) \hbox{d} \tau \right).
\end{eqnarray}

\subsubsection{Tractable formulation of the expected deviance}
By definition of the MGF $m$, the expected deviance $\Phi_1$ is the derivative of $m$ with respect to $\beta$ evaluated at $\beta=0$,  
\begin{eqnarray*}
\Phi_1(\alpha) &=& \frac{\partial m(\alpha, \beta)}{\partial \beta}\Big|_{\beta=0}  \\
 &=&  \frac{\partial \mathbb{E}\left[ p^{\alpha+\beta}(\bm{Y}|\bm{\theta}) \right]}{\partial \beta}  \Big|_{\beta=0}  \, \exp \left(-\int_0^\alpha \Phi_1(\tau) \hbox{d} \tau \right)\\
 &=& \mathbb{E}\left[ \log p(\bm{Y}|\bm{\theta}) \,p^{\alpha+\beta}(\bm{Y}|\bm{\theta}) \right] \Big|_{\beta=0} \exp \left(- \int_0^\alpha \Phi_1(\tau) \hbox{d} \tau \right)\\
 &=& \mathbb{E}\left[ \log p(\bm{Y}|\bm{\theta}) \,p^{\alpha}(\bm{Y}|\bm{\theta}) \right]\exp \left(- \int_0^\alpha \Phi_1(\tau) \hbox{d} \tau \right).
\end{eqnarray*}
We let $h$ denote the tractile function depending on the tempering variable
\begin{eqnarray*}
h(\alpha) =  \mathbb{E}\left[ \log p(\bm{Y}|\bm{\theta}) \,p^{\alpha}(\bm{Y}|\bm{\theta}) \right],
\end{eqnarray*}
the expected deviance $\Phi_1$ is therefore solution of the nonlinear integral equation
\begin{eqnarray}
\label{eq:ed1}
\Phi_1(\alpha)= h(\alpha) \exp \left(- \int_0^\alpha \Phi_1(\tau) \hbox{d} \tau \right).
\end{eqnarray}
At this stage, numerical methods for solving integral equations can be applied to \eqref{eq:ed1} to compute efficiently  the expected deviance in a tempering grid, $0=\alpha_0 \leq \alpha_1 \leq \cdots \leq \alpha_{N_{_\alpha}}=1$, of $N_{_\alpha}+1$ points. Nevertheless, we push forward towards a closed form of $\Phi_1$
%\paragraph{\underline{Soving the nonlinear equation}}: We are addressing here the analytical resolution of the nonlinear equation \eqref{eq:ed1}. 
by differentiating \eqref{eq:ed1} with respect to the tempering variable, by writing
\begin{eqnarray*}
\frac{\partial \Phi_1}{\partial \alpha}(\alpha) 
&=& \frac{\partial h(\alpha)}{\partial \alpha}\exp \left(-\int_0^\alpha \Phi_1(\tau) \hbox{d} \tau \right) -  h(\alpha) \Phi_1(\alpha) \exp \left(-\int_0^\alpha \Phi_1(\tau) \hbox{d} \tau \right)\\
&=& \frac{\partial h(\alpha)}{\partial \alpha}\frac{1}{h(\alpha)} \left( h(\alpha) \exp \left(- \int_0^\alpha \Phi_1(\tau) \hbox{d} \tau \right)\right) -  \Phi_1(\alpha) \left( h(\alpha)  \exp \left(- \int_0^\alpha \Phi_1(\tau) \hbox{d} \tau \right) \right)\\
&=& \frac{\partial h(\alpha)}{\partial \alpha}\frac{1}{h(\alpha)}  \Phi_1(\alpha)  -  \Phi_1(\alpha) \Phi_1(\alpha)  \\
&=&  \frac{\partial h(\alpha)}{\partial \alpha}\frac{1}{h(\alpha)}  \Phi_1(\alpha)  -  \Phi_1^2(\alpha).
\end{eqnarray*}
We find that the expected deviance is governed by a differential Bernouilli equation of second-type given by
\begin{eqnarray}
\left\{
\begin{array}{lll}
\Phi_1'(\alpha) =\frac{h'(\alpha)}{h(\alpha)}   \Phi_1(\alpha)  - \Phi_1^2(\alpha)\quad \hbox{for}\quad 0<\alpha\leq 1,\\\\
\Phi_1(0) = \mathbb{E}\left[ \log p(\bm{Y}|\bm{\theta}) \right],
\end{array}
\right.
\end{eqnarray}
with solution in the form
\begin{eqnarray*}
\Phi_1(\alpha) = \frac{h(\alpha)}{1 + \displaystyle{ \int_0^\alpha h(\tau) \hbox{d} \tau }}.
\end{eqnarray*}
To gain clarity, we summarize the analysis presented in Section \ref{sec:solmgf} in the following theorem.
\begin{theorem}
The power-posterior distribution $\pi(\bm{\theta}|\bm{Y};\alpha)$ exists and is given by
\begin{eqnarray}
\label{eq:soltract1}
   \pi(\bm{\theta}|\bm{Y};\alpha) = \pi(\bm{\theta})  p^\alpha (\bm{Y}|\bm{\theta}) \exp\left(-\int_0^\alpha \mathbb{E}_{\bm{\theta}|\bm{Y}; \tau}\left[\log p(\bm{Y}|\bm{\theta}) \right] \hbox{d} \tau \right)\quad \hbox{for}\quad \alpha \in [0,1],
\end{eqnarray}
where the expected deviance is 
\begin{eqnarray}
\label{eq:tractE}
\mathbb{E}_{\bm{\theta}|\bm{Y}; \alpha}\left[\log p(\bm{Y}|\bm{\theta}) \right] = \frac{h(\alpha)}{1 + \displaystyle{ \int_0^\alpha h(\tau) \hbox{d} \tau }}.
\end{eqnarray}
\end{theorem}
 Computing the power posterior distribution demands forward evaluations only for approximating $\Phi_1(0)$. The tractability, through  \eqref{eq:tractE}, provides the expected deviance  $\mathbb{E}_{\bm{\theta}|\bm{Y}; \alpha}\left[\log p(\bm{Y}|\bm{\theta}) \right]= \Phi_1(\alpha)$ at any stage without any extra evaluations of the forward model $\bm{g}$. In the next section, we justify the setting of \eqref{eq:tractE} in term of computational stability and perform numerical experiments.

\section{Numerical results}\label{sec:num}
We conduct numerical results to assess the applicability and the practical aspect of our approach of computing the transition distributions while holding back evaluations of the forward model at the prior stage. Notations $\mathcal{U}$, $\mathcal{N}$ and $\Gamma$ refer to the uniform, normal and Gamma distributions, respectively. We present three numerical examples, the first is a one-dimensional linear model with Gaussian measurement errors where an explicit form of the expected deviance and of the tractile function $h$ allow to perform a detailed error analysis. The second example deals with a source inversion problem where the measurement errors are in a more realistic setting and are incorporated in the likelihood function  using the Wasserstein distance. The third example is a challenging problem of multimodal inversion of a two-dimensional parameter.

\subsection{Scaling domain and implementation}
In the case of no interest in the transition distributions, the integral at \eqref{eq:solppd} sums from $0$ to $1$ and can be casted as the expectation, with respect to the joint distribution $\pi(\bm{\theta}, \alpha)$ of $\left(\bm{\theta}, \alpha\right)$, of the ratio of the log-likelihood over the distribution of the tempering parameter $\mathbb{E}_{\bm{\theta}, \alpha|\bm{Y}} \left[\frac{\log p(\bm{Y}|\bm{\theta})}{p(\alpha)}\right]$.  The total computational bias of that formulation copies only the model bias, generated for evaluating the forward model $\bm{g}$, but requires the distribution $p(\alpha)$ as given. The design of such unbiased estimators with MCMC methods, to compute the normalizing constant, has been addressed in  \cite{Jacob2019, Rischardetal2018}.

In the second instance, one would look at the expected deviance as the ratio of the expected values with respect to the prior PDF; $\Phi_1(\alpha) = \frac{\mathbb{E} \left[ p^\alpha(\bm{Y}|\bm{\theta}) \log p(\bm{Y}|\bm{\theta})\right]}{\mathbb{E}\left[ p(\bm{Y}|\bm{\theta})\right]}$. As the numerator and the denominator have different scales of the likelihood, that formulation is likely unstable numerically because its ratio scale promotes numerical underflow even for small likelihood function ($p(\bm{Y}|\bm{\theta}) \leq 1$).

\begin{figure}[H]
\centering
\includegraphics[width=0.50\textwidth]{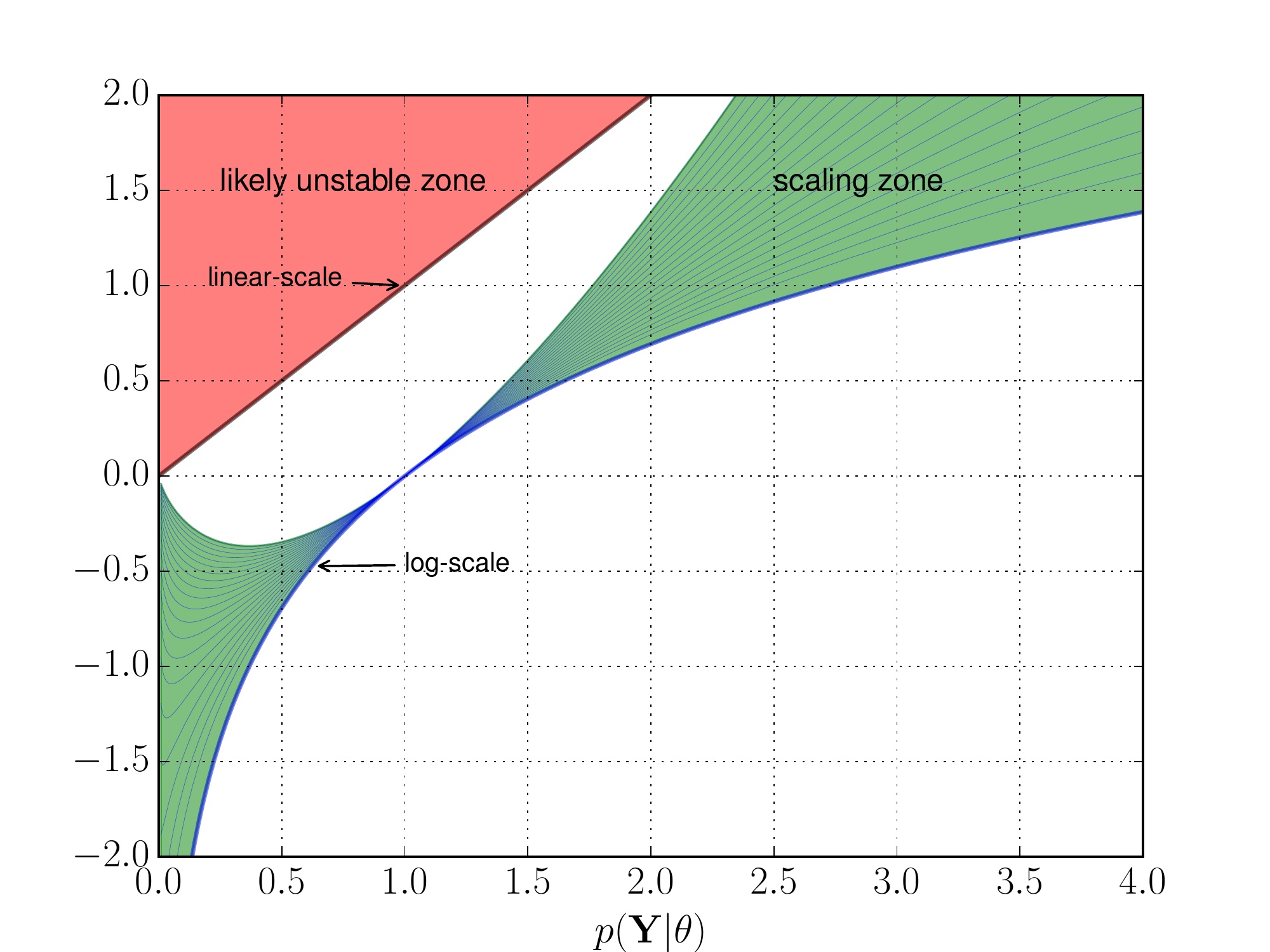}
\caption{scaling domain}
\label{fg:scaling}
\end{figure}
A crucial asset of the thermodynamic identity is that the expected deviance is on the logarithmic scale of the likelihood, which provides an unconditional numerical stability. That feature has been explored to advance MCMC methods in many ways in devising paths to access the posterior PDF (or in other terms in approximating the Bayes factor), see for instance \cite{GelmanMeng1998, Neal1996}. We believe that the logarithmic scale remains the best in terms of numerical stability despite the intractability issue and the intensive computational cost when it comes to profile transition distributions.
One main purpose of this work is to avoid evaluations of the forward model $\bm{g}$ at transition stages, that subsequently induces the computational tractability. We also seek to preserve the numerical stability in the computation of the expected deviance. Rather than adopting a fixed scale of the likelihood function, our approach sweeps a scaling zone when computing the expected deviance. The scale changes from the log scale of the likelihood to the linear-log scale of the likelihood when $\alpha$ varies from $0$ to $1$. Although we fabricate our approach wih the assymption of $p(\bm{Y}|\bm{\theta}) \leq 1$ at a number of steps in the theoretical derivation, it is worth to mention that when $\alpha$ approaches $1$ with large likelihood function ($p(\bm{Y}|\bm{\theta}) \gg 1$), the ratio scale of the expected deviance reduces drastically the risk of numerical instability due to the fact that the numerator and the denominator have the same scale in the likelihood function. In Figure \ref{fg:scaling}, we schematize the scaling zone sweeping by the computation of the expected deviance $\Phi_1(\alpha)$ for $\alpha \in [0,1]$, the likely unstable zone, and the log scale. 
\begin{center}
\begin{minipage}{0.59\textwidth}
\footnotesize{
\begin{algorithm}[H]\label{algo1}
\DontPrintSemicolon 
  \KwInput{\vspace{-0.5cm}
   \begin{align*}
   &  \hbox{Data} \quad \bm{Y}\\
   & \hbox{forward model} \quad \bm{g}\\
   & \hbox{tempering gridpoints} \quad \{\alpha_0, \alpha_1, \cdots, \alpha_{N_{_{\tiny{\alpha}}}} \}
\end{align*}  
  }
  \vspace{-0.3cm}
  \KwOutput{$\Phi_1(\alpha_k)$ for $k=0,\cdots,{N_{_{\tiny{\alpha}}}}$} %\vspace{0.3cm}
  Generate  $\bm{\theta}_1, \cdots, \bm{\theta}_{_N}$ from $\pi(\bm{\theta})$ \\ %\vspace{0.3cm}  
  Evaluate $\bm{g}(\bm{\theta}_1), \cdots, \bm{g}(\bm{\theta}_{_N})$ \\ %\vspace{0.3cm}
  Evaluate $\log p(\bm{Y}|\bm{\theta}_n)$ for $n=1,\cdots,N$\\ %\vspace{0.3cm}
  Approximate $h(\alpha_0) := \mathbb{E}\left[\log p(\bm{Y}|\bm{\theta}) \right]$, set $\hbar(\alpha_0) := 0$ \\ %\vspace{0.3cm}
  \For{k in $\{1, \cdots,  {N_{_{\tiny{\alpha}}}} \}$} 
  {
     Evaluate $p^{\alpha_k}(\bm{Y}|\bm{\theta}_n)$ for $n=1,\cdots,N$\\ %\vspace{0.3cm}
     Approximate $h(\alpha_k) := \mathbb{E}\left[p^{\alpha_k}(\bm{Y}|\bm{\theta}) \log p(\bm{Y}|\bm{\theta}) \right]$\\ %\vspace{0.3cm}
     Approximate $\hbar(\alpha_k) := \hbar(\alpha_{k-1}) + \displaystyle{ \int_{\alpha_{k-1}}^{\alpha_k}h(\tau) \hbox{d} \tau }$\\ %\vspace{0.1cm}
     Evaluate $\Phi_1(\alpha_k) := \displaystyle{\frac{h(\alpha_k)}{1+\hbar(\alpha_k)}}$
  }
\caption{\footnotesize{Computation of $\Phi_1$}}
\end{algorithm}}
\end{minipage}
\end{center}
The Algorithm \ref{algo1} summarizes the numerical approximation of the expected deviance for a given grid $\left(\alpha_k\right)_{0\leq k \leq N_\alpha} $. Only $N$ evaluations of $\bm{g}$ at the prior stage suffice to compute the expected deviance $\Phi_1$ at any stage $\alpha_k$, including the posterior one ($k = N_\alpha$). That induces the tractability of the power posterior PDFs. In this work, the expected values at steps 4 and 7 in algorithm \ref{algo1} are approximated using the Monte Carlo sampling, while the Simpson quadrature is employed for the local integration, from $\alpha_{k-1}$ to $\alpha_k$, of the tractile function $h$ at step 8. Next, we present numerical experiments that cover only the method bias (bias generated by the tempering quadrature at step 8 in algorithm \ref{algo1}) since the evaluation of the forward model $\bm{g}$, in each of the three examples,  does not require a numerical discretization. 

\subsection{Example 1: A linear algebraic model}
We consider an algebraic one-dimensional linear model, where the expected deviance $\Phi_1$ and the tractile function $h$ can be computed analytically, given by
\begin{equation}
\label{eq:linearmodel}
y = A\theta + \epsilon.
\end{equation}
Here our prior belief in the parameter of interest is characterized by a normal PDF,  $\theta \sim \mathcal{N}\left(1, \sigma_p^2\right)$ and we consider that the observational noise for the data is Gaussian as well, i.e. $\epsilon \sim \mathcal{N}\left(0, \sigma_\epsilon^2\right)$, which leads to a likelihood function of the form $p(y|\theta) = (2\pi \sigma_\epsilon^2)^{-1/2} \exp \left(-\frac{1}{2}\frac{(y - A\theta)^2}{\sigma_\epsilon^2}\right)$. Let $\sigma_\alpha^2 = A^2 \sigma_p^2\alpha + \sigma_\epsilon^2$, we find that the expected deviance and the tractile function are, respectively, given by
\begin{eqnarray*}
 \mathbb{E}_{\theta \vert y; \alpha} \left[ \log p(y \vert \theta) \right]  &=& -\frac{1}{2} \left( \log \left(2\pi \sigma_\epsilon^2\right) + \frac{A^2 \sigma_p^2}{\sigma_\alpha^2} + \frac{\left(y - A\right)^2 \sigma_\epsilon^2}{\sigma_\alpha^4}  \right),\\
 h(\alpha) &=& -\frac{1}{2}  \left(2\pi \sigma_\epsilon^2\right)^{-\alpha/2}    \left( \log \left(2\pi \sigma_\epsilon^2\right) + \frac{A^2 \sigma_p^2}{\sigma_\alpha^2} + \frac{\left(y - A\right)^2 \sigma_\epsilon^2}{\sigma_\alpha^4}  \right)  \exp\left(-\frac{\alpha}{2} \frac{ (y- A)^2}{\sigma_\alpha^2}\right).
 \end{eqnarray*}
The synthetic data sample is acquired from the "\textit{true}" value $\theta_t = 1.1$ with measurement errors diffusion coefficient of $\sigma_\epsilon^2 = 4.0$.
\begin{figure}[H]
\centering
\includegraphics[width=0.7\textwidth]{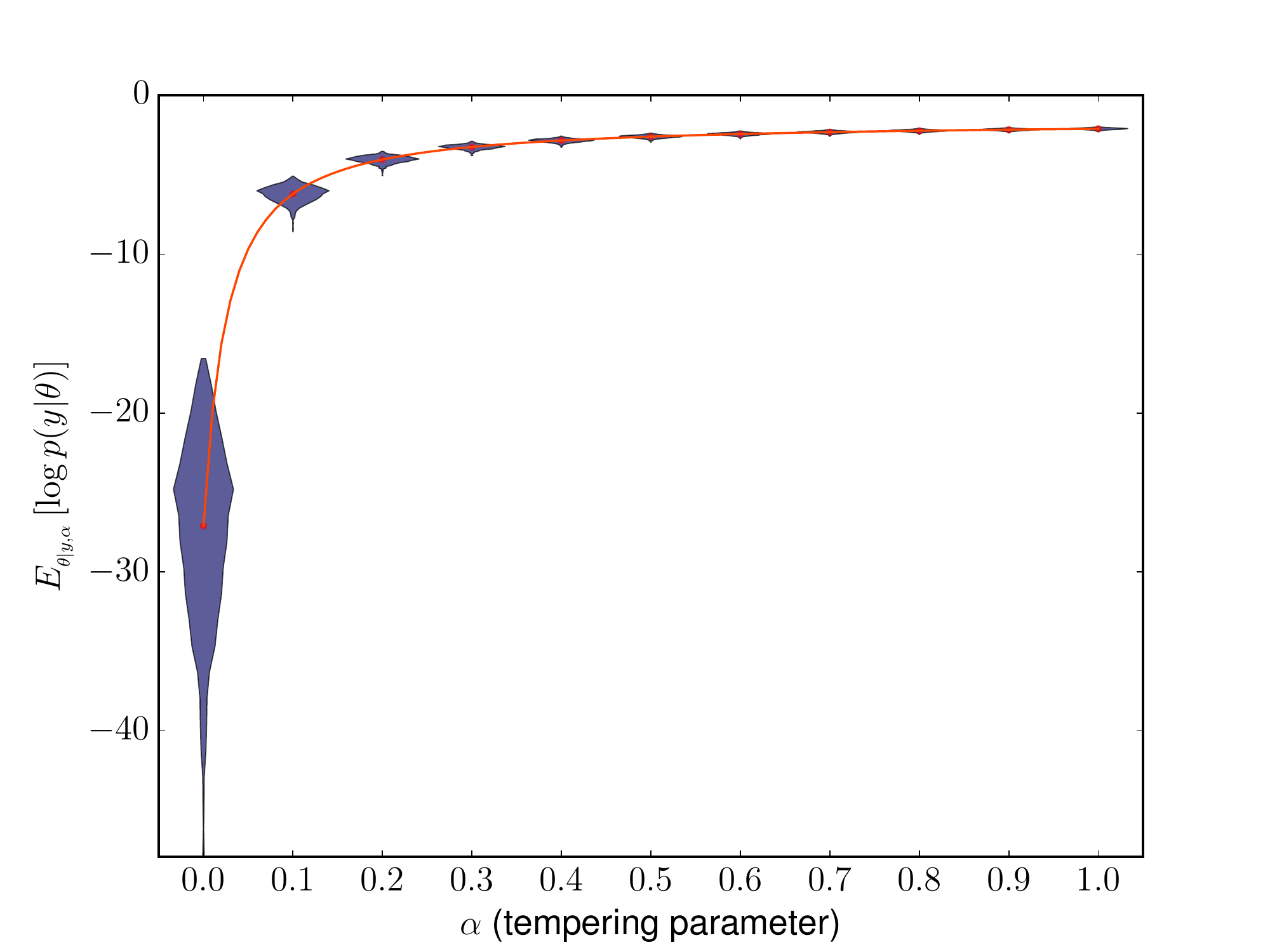}
\caption{(Example 1) Distribution of the expected deviance and its expectation with respect to the evidence for the data model \eqref{eq:linearmodel} with $A = 101$ and $\sigma_\epsilon^2 = 4.0$,  $\sigma_p^2 = 10^{-2}$.}
\label{fg:Ex1deviance}
\end{figure}
In Figure \ref{fg:Ex1deviance}, we represent the expected deviance in term of the tempering parameter, and its expectation with respect to a data sample of size $10^3$, where the tractile function $h$ is approximated by Monte Carlo sampling with $10^3$ samples. We also plot the distribution of the expected deviance at ten tempering points. The main observation here is the important variance reduction from the prior ($\alpha =0$) to the posterior ($\alpha =1$). As statement of this significant outcome, we observe that it is now possible to accurately estimate statistics with very few data, under the posterior distribution, with evaluations of $\bm{g}$ performed at the prior stage. 
\begin{figure}[H]
\centering
\includegraphics[width=1.0\textwidth]{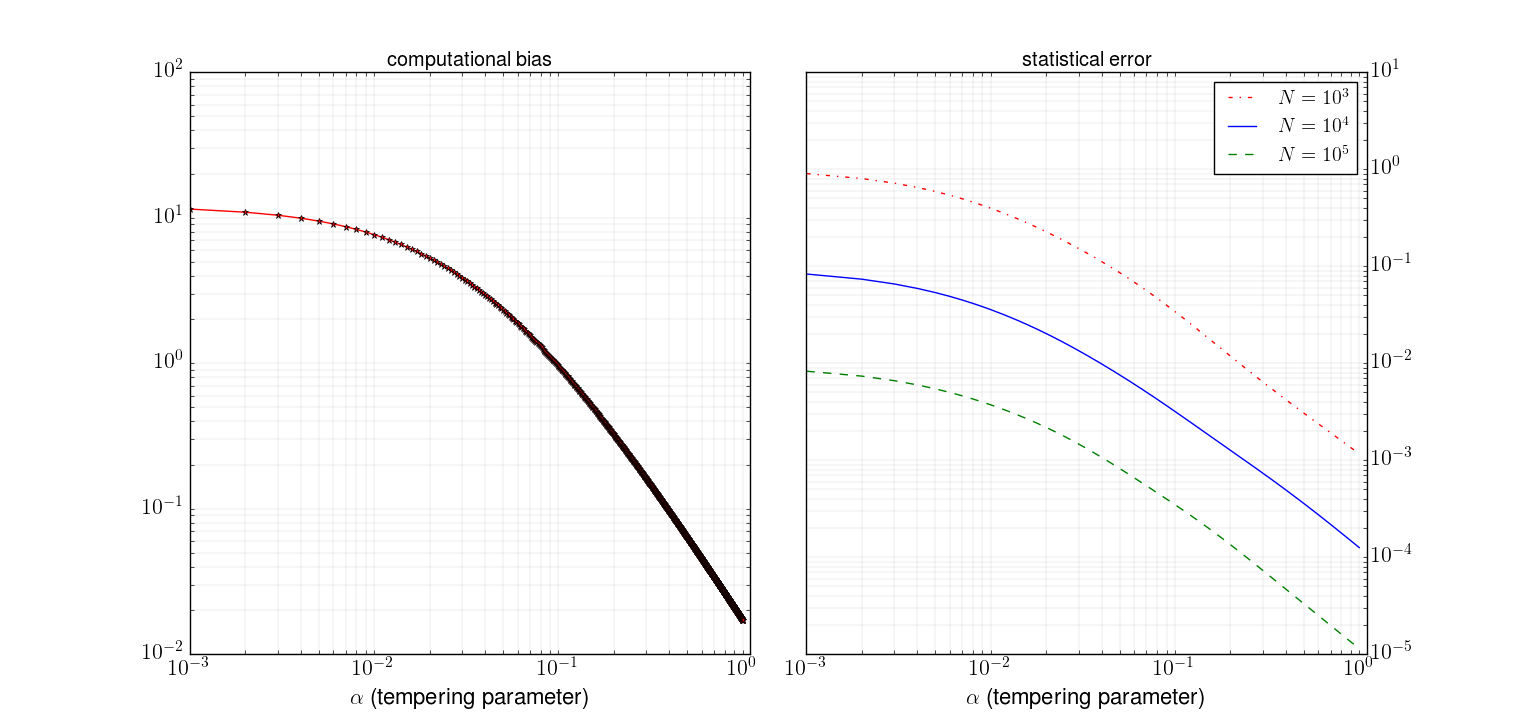}
\caption{(Example 1) Error analysis}
\label{fg:Ex1error}
\end{figure}
Figure \ref{fg:Ex1error} depicts the variations of the computational bias with respect to the tempering parameter and the statistical error for differents $N$, number of random variables generated from the prior, in the Monte Carlo sampling to estimate $h(\alpha)$. The computational bias is approximated as the mean value over a dataset of size $10^3$ of the absolute value of the exact $\mathbb{E}_{\theta \vert y; \alpha} \left[ \log p(y \vert \theta) \right]$ minus the expected deviance at \eqref{eq:tractE} evaluated with the exact tractile function $h(\alpha)$. Here, the bias is purely associated to the tempering quadrature error, and is improving by order $10^{-3}$ from the prior to the posterior. We let the statistical error be the $L^2$ norm of the difference between the exact deviance and the expected deviance \eqref{eq:tractE} evaluated with  $h_{_N}(\alpha)$, where  $h_{_N}(\alpha)$ is the Monte Carlo approximation of the tractile function with $N$ samples. The statistical error plot (right-plot, Figure \ref{fg:Ex1error}) allows to quantify the variance reduction observed in Figure \ref{fg:Ex1deviance}. From each of the three cases, $N = 10^{3}, 10^{4}, 10^{5}$, one can see a significant variance reduction of the expected deviance from $10^3$ at the prior stage to $1$ at the posterior one.

\subsection{Example 2: Source inversion with Wasserstein metric-based likelihood}
We now consider the example of source inversion problem, treated in \cite{Motamed2019},  where an initial wave pulse propagates with a constant speed. The mathematical governing equation is the one-dimensional wave equation
\begin{eqnarray}
\label{eq1}
\left\{
\begin{array}{ll}
\displaystyle{u_{tt}(t,x) - u_{xx}(t,x) = 0 } \; & \displaystyle{\text{in}\quad ]0,T]\times \mathbb{R}, }
\vspace{0.1cm} \\
\displaystyle{u(0,x) = u^0(x, \bm{\theta})} \; & \displaystyle{\text{in}\quad  \mathbb{R}, } 
\vspace{0.1cm} \\
\displaystyle{u_t(0,x) = 0} \; & \displaystyle{ \text{in}\quad \mathbb{R}. }
\end{array}
\right.
\end{eqnarray}
We aim to assess how our continuation inversion approach performs for complicated structure of the measurement errors composed with a multiplicative and additive noises. We consider an initial pulse of the form
\begin{eqnarray*}
u^0(x, \bm{\theta}) = a\exp \left(-100(x-x_0-0.5)^2 \right)+  a \exp \left(-100(x-x_0)^2 \right) + a \exp \left(-100(x-x_0+0.5)^2 \right)
\end{eqnarray*}
from an unknown location $x_0$ with an unknown amplitude $a$ that we aim to learn continuously from the prior stage to the posterior one. Therefore the parameter of interest is $\bm{\theta} = (x_0, a)$.  The forward model $g$ is the solution of \eqref{eq1} that is in the form
\begin{eqnarray}
 g(t,x, \bm{\theta}) \myeq  u(t,x) = \frac{1}{2}u^0(x-t,\bm{\theta}) + \frac{1}{2}u^0(x+t,\bm{\theta}).
\end{eqnarray}
Seven receivers, $N_r = 7$, are placed at locations $x_1 = -3$, $x_2 = -2$, $x_3 = -1$, $x_4 = 0$,  $x_5 = 1$, $x_6 = 2$ and $x_7 = 3$ to record polluted discrete-time signal $\bm{y}_r$ over the time interval $[0,T]$ at $N_{_T}$ time-steps $t_k = \frac{k-1}{N_{_T} -1} T$, $ k=1, \cdots, N_{_T}$. The data $\bm{Y}$ is represented as $\bm{Y} = \left\{ \bm{y}_1, \cdots , \bm{y}_{N_r}\right\}$, where the $r^{th}$ receiver is counting the responses $y(t_k,x_r)$ at all the $N_{_T}$ time-steps, $\bm{y}_r = \left(y(t_k,x_r)\right)_{1 \leq k \leq N_{_T}}$ where
\begin{eqnarray}
 y(t_k,x_r) = \epsilon_{kr}^{(1)} g(t_k,x_r, \bm{\theta}_t) + \epsilon_{kr}^{(2)},
\end{eqnarray}
with $\epsilon_{kr}^{(1)} \sim \Gamma(60,1/60)$ and $\epsilon_{kr}^{(2)} \sim \mathcal{U}(-0.25,0.25)$. In the numerical experiment, we consider $T=5$.
\begin{figure}[H]
\centering
\includegraphics[width=0.70\textwidth]{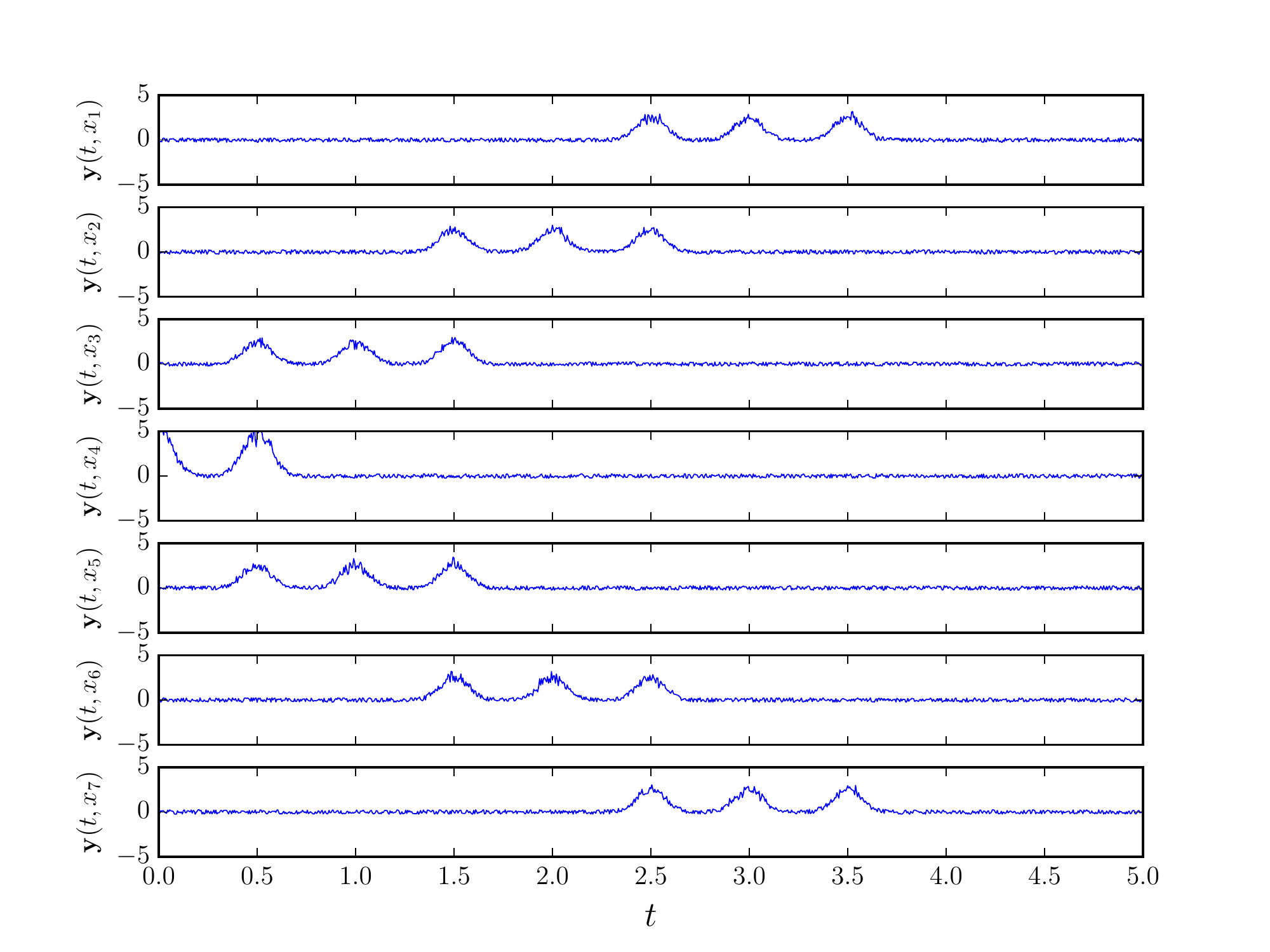}
\caption{(Example 2) Recorded signal $\bm{Y}$ at the seven receivers with $N_{_T}=1001$.}
\label{fg:Ex2signal}
\end{figure}
To compute the power posterior distributions of $\bm{\theta}$, we consider the likelihood function as follows
\begin{eqnarray}
p(\bm{Y}|\bm{\theta}) = s^{N_{_T}} \exp \left(-s \sum_{r=1}^{N_r} \dis_{_W}\left(\bm{y}_r, \bm{g}_r(\bm{\theta})\right) \right),
\end{eqnarray}
where $\bm{g}_r(\bm{\theta})$ gathers the $N_{_T}$ outputs of the forward model, $\bm{g}_r(\bm{\theta}) = \left(g(t_k,x_r,\bm{\theta})\right)_{1 \leq k \leq N_{_T}}$ evaluated with the random parameter $\bm{\theta}$, and $\dis_{_W}$ stands for the quadratic Wasserstein distance. The errors dispersion parameter $s$ is also not available to us and is assumed random, $s \sim \Gamma(1, 0.1)$.

\begin{figure}[H]
\centering
\includegraphics[width=0.7\textwidth]{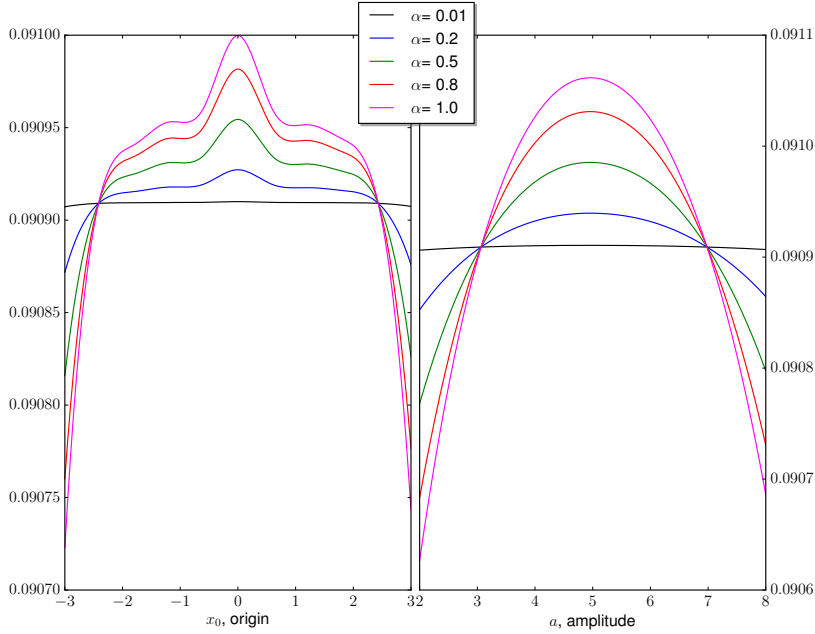}
\caption{(Example 2) Marginal transition distributions of the origin $x_0$ (left-plot) and the amplitude $a$ (right-plot) conditioned with a recorded signal of length $N_{_T}=101$  at the seven receivers.}
\label{fg:Ex2pdfs}
\end{figure}
Figure \ref{fg:Ex2pdfs} presents the marginal power posterior distributions of the location $x_0$ (left plot) and of the amplitude $a$ (right plot) at different tempering points.  The simulation is performed with a single synthetic data generated with the \textit{"true"} location and \textit{"true"}pulse amplitude  $\bm{\theta}_t = (0.0, 0.5)$, with signal length of $N_{_T} = 101$. 
The tractile function $h$ is approximated with the Monte Carlo sampling with $10^2$ samples.  Our continuation method handles perfectly that complex noise structure as the marginal power posteriors concentrate around the true values when the tempering parameter goes to $1$.

\subsection{Example 3: Bivariate multimodal inversion}
In this example we confront our approach, in terms of capturing the evolution of the distribution from the prior to the posterior, to a mixture of $M$ two-dimensional Gaussians. Let us consider the data representation
\begin{eqnarray}
 \bm{Y} = \bm{\theta} + \bm{\epsilon},
\end{eqnarray}
where the measurement errors is Gaussian, $\bm{\epsilon} \sim \mathcal{N}\left(\bm{0}, \sigma^2 \bm{I}_2\right)$, $\bm{I}_2$ is the $2\times 2$ unit matrix, and $\bm{\theta}$ is drawn form the bivariate uniform prior, $\bm{\theta} \sim \mathcal{U}\left([0,10]\right) \times \mathcal{U}\left([0,10]\right)$. We consider a synthetic data, of size $N_{_{\tiny{\hbox{data}}}}$, collected from the categorical distribution of $M$ modes, $\bm{Y} \sim \mathcal{N}\left(\bm{\theta}_t, \sigma^2 \bm{I}_2\right)$ with $\bm{\theta}_t = \bm{\mu}_1 \vee \bm{\mu}_2 \vee \cdots \vee \bm{\mu}_{M}$, where $\vee$ stands for the logical connective operator "\textit{or inclusive}". This example  is built on the multimodal simulation case treated for assessing advance classes of MCMC methods \cite{Lian2001, Stojkova2017}. The likelihood function has the form
\begin{eqnarray}
p(\bm{Y}|\bm{\theta}) = \prod_{m=1}^{M} p^{[w(\bm{y}_m)= w_m]}(\bm{y}_m|\bm{\theta}) %\times p^{[w(\bm{y}_2)= w_2]}(\bm{y}_2|\bm{\theta})  \times p^{[w(\bm{y}_3)= w_3]}(\bm{y}_3|\bm{\theta}) 
\end{eqnarray}
for a data $\bm{Y} = (\bm{y}_1, \cdots, \bm{y}_M)$ sampled from $\bm{Y} \sim \mathcal{N}\left(\bm{\theta}_t, \sigma^2 \bm{I}_2\right)$, where $[\cdot=\cdot]$ represents the Iverson bracket, $w(\bm{y}_m)$ is the weight of $\bm{y}_m$, and $w_m$ is the assigned weight to the mode $\bm{\mu}_m$, $m=1,\cdots, M$.

\begin{figure}[H]
\centering
\includegraphics[width=0.5\textwidth]{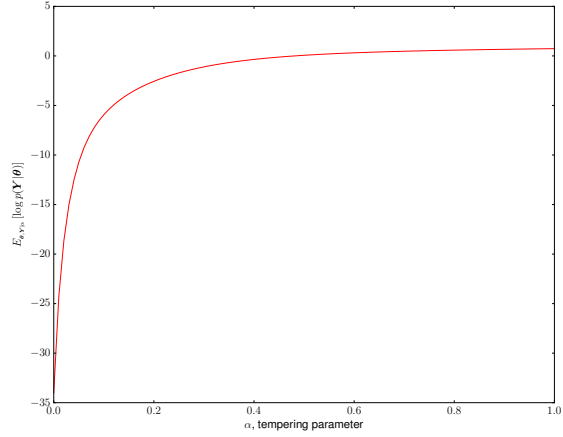}
\caption{(Example 3) Expected deviance with respect to joint distribution of the parameter of interest and the data.}
\label{fg:Ex3deviance}
\end{figure}
To demonstrate further the effectiveness of our proposed method, we investigate the ability of computing the expected deviance with respect to the joint distribution of the parameter of interest and the data. For $N_{_{\tiny{\hbox{data}}}} = 10^3$, $N=10^2$, we show in Figure \ref{fg:Ex3deviance} the evolution of the expected deviance in terms of the tempering parameter. The categorical distribution has $M=3$ modes in this numerical experiment with  $\bm{\mu}_1 = (2.18, 5.76)$, $\bm{\mu}_2 = (8.41, 1.68)$ and $\bm{\mu}_3 = (5.54, 6.86)$.

\begin{figure}[H]
\centering
\includegraphics[width=0.7\textwidth]{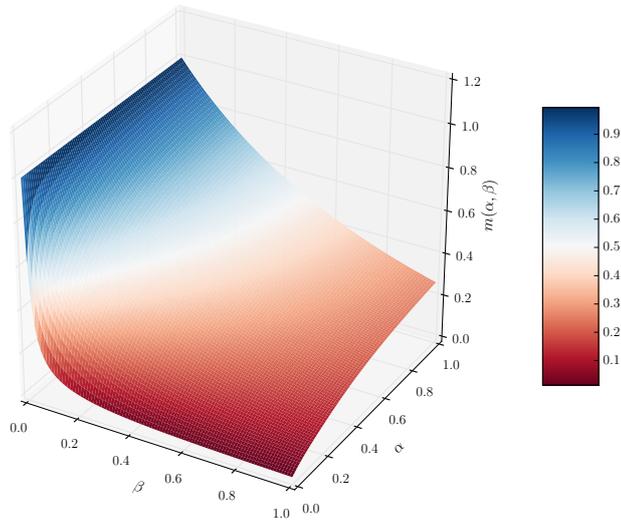}
\caption{(Example 3) Moments generating function surface in the unit square $[0,1]\times[0,1]$.}
\label{fg:Ex3mgf}
\end{figure}
Using the same configuration: categorical distribution with three modes,  $N_{_{\tiny{\hbox{data}}}} = 10^3$, $N=10^2$, we display the MGF for $(\alpha,\beta) \in [0,1]\times[0,1]$ to stay consistent  with the fact
the MGF only exists for $\beta \in [0, 1]$ and as our interest to the MGF is for small $\beta$.

In multimodal inversion, one important observation when using the Bayesian approach is that all the modes might not be captured for small size datasets.  We investigate this phenomena using our continuation approach by setting two different configurations. The first is a weighted categorical distribution with $M=3$ modes, where $w_1 = 0.3$, $w_2 = 0.5$ and $w_3=0.2$ (Figure \ref{fg:Ex3posteriorPDFs1}).  The second configuration deals with a categorical distribution with $M=20$, where $w_m=1/20$, for $m=1,\cdots,20$ and the modes locations are collected from \cite{Lian2001}, where they are uniformly sampled from the prior $\pi(\bm{\theta})$ and given by
\begin{align*}
&\bm{\mu}_1 = (2.18, 5.76), & \bm{\mu}_6 = (3.25, 3.47),&   & \bm{\mu}_{11} = (5.41, 2.65), & &\bm{\mu}_{16} = (4.93, 1.50),\\
&\bm{\mu}_2 = (8.67, 9.59), & \bm{\mu}_7 = (1.70, 0.50),&   & \bm{\mu}_{12} = (2.70, 7.88), & &\bm{\mu}_{17} = (1.83, 0.09),\\
&\bm{\mu}_3 = (4.24, 8.48), & \bm{\mu}_8 = (4.59, 5.60),&   & \bm{\mu}_{13} = (4.98, 3.70), & &\bm{\mu}_{18} = (2.26, 0.31),\\
&\bm{\mu}_4 = (8.41, 1.68), & \bm{\mu}_9 = (6.91, 5.81), &  & \bm{\mu}_{14} = (1.14, 2.39), & &\bm{\mu}_{19} = (5.54, 6.86),\\
&\bm{\mu}_5 = (3.93, 8.82), & \bm{\mu}_{10} = (6.87, 5.40),& & \bm{\mu}_{15} = (8.33, 9.50), & &\bm{\mu}_{20} = (1.69, 8.11).
\end{align*}
Figures \ref{fg:Ex3posteriorPDFs1} and \ref{fg:Ex3posteriorPDFs2} are obtained with a Monte Carlo approximation of $N=10^2$ samples of the tractile function $h$, and Simpson quadrature with $11$ gridpoints is used for the tempering quadrature between $\alpha_k$ and $\alpha_{k+1}$.

\begin{figure}[H]
\centering
\includegraphics[width=1.09\textwidth]{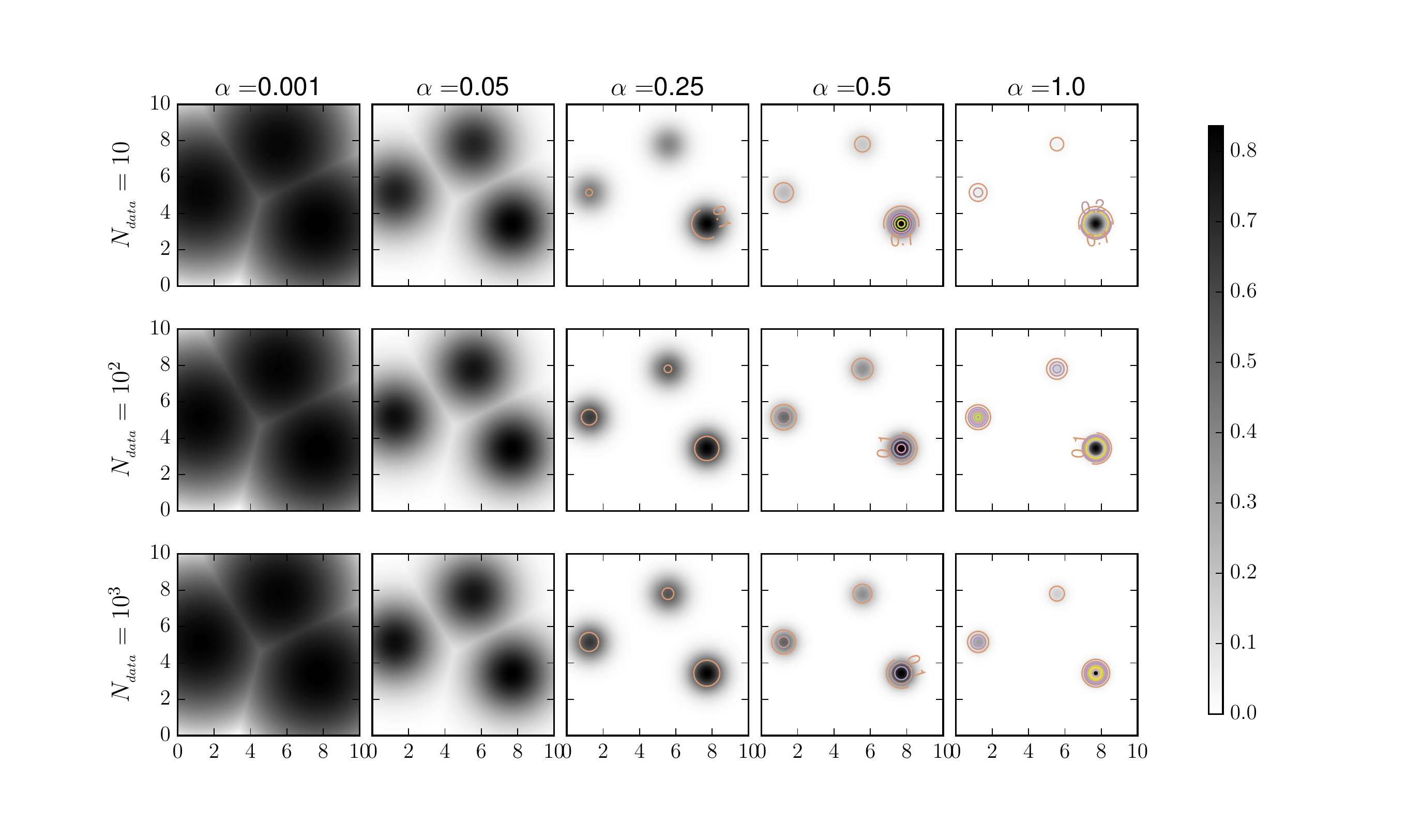}
\caption{(Example 3) Power posteriors distributions with three modes}
\label{fg:Ex3posteriorPDFs1}
\end{figure}
\begin{figure}[H]
\centering
\includegraphics[width=1.09\textwidth]{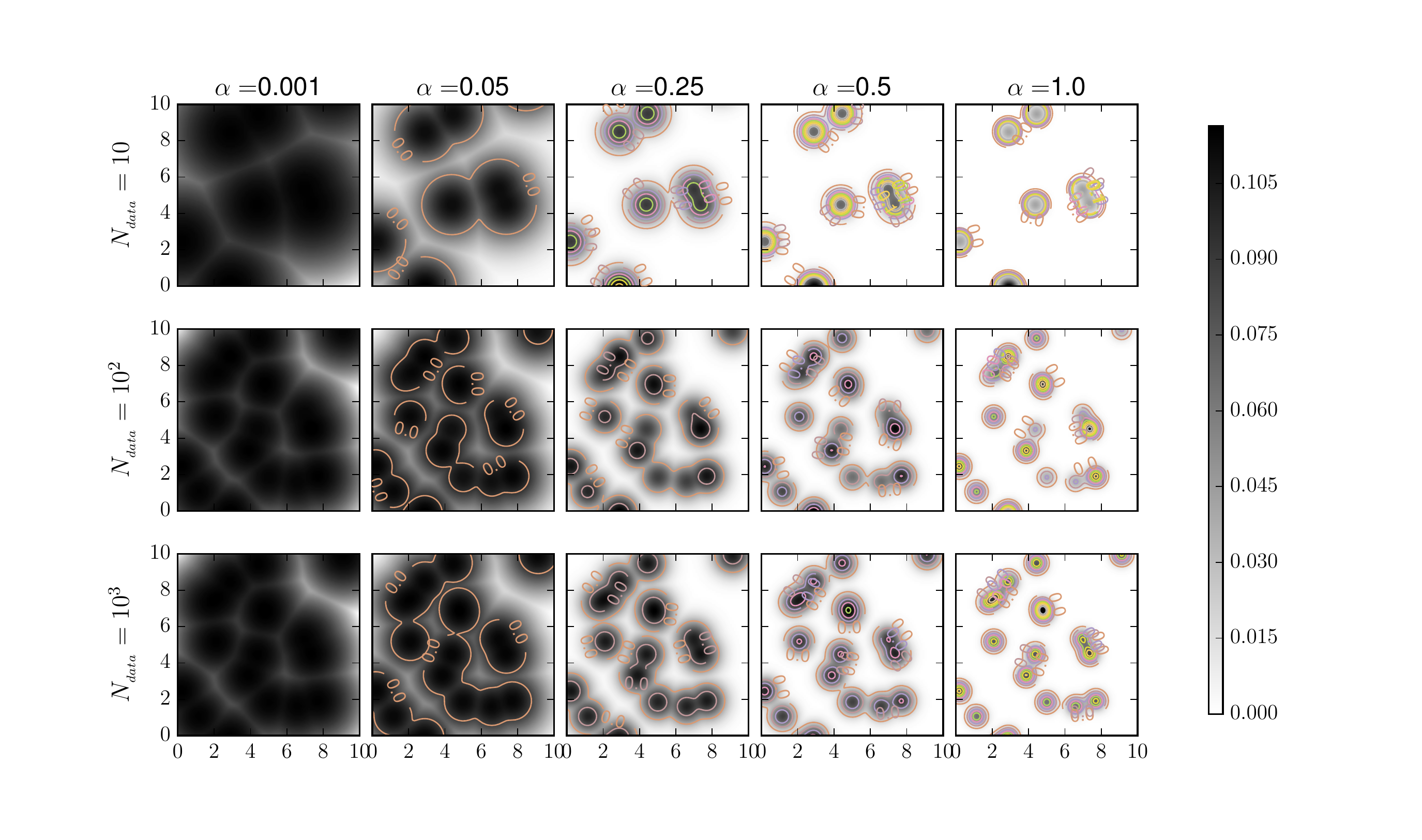}
\caption{(Example 3) Power posteriors distributions with 20 modes}
\label{fg:Ex3posteriorPDFs2}
\end{figure}
In Figures \ref{fg:Ex3posteriorPDFs1} and \ref{fg:Ex3posteriorPDFs2}, we visualize the update of the distribution of the parameter of interest. For the first case, $M=3$, the power posterior distribution shows all the three modes at the earlier stage even for the small size dataset, and the distribution concentrates around the modes as the tempering parameter goes to $1$. In Figure \ref{fg:Ex3posteriorPDFs2}, our continuation approach exhibits, at the earlier stage, the effect of non-capturing all the modes  for small size datasets. However, the number of captured modes increased with the dataset for the largest size.

\section{Conclusions}\label{sec5}
We presented a continuation method for interpreting the posterior probability density function, as the terminal stage of a path of transition distribution functions between the prior and the posterior, for Bayesian inverse problems with exponential likelihood function. A mathematical analysis for the existence and the uniqueness of the transition distributions has been presented to support the consistency of our results. We managed to hold back all the evaluations of the forward model at the prior stage in the approximation of the transition distribution. That achieves the computational tractability of the posterior distribution in opposite to classical MCMC methods where the acceptance condition complicates the counting, a priori, of the number the forward model evaluations. The computational tractability was addressed through the tractability of the expected deviance using the moments generating function of the log-likelihood function. Our research has stressed the important potential of continuation methods for two of the three classes of uncertainty quantification tasks enumerated in \cite{Rude2018}.

The presented numerical results dealt with evaluations of simple forward models while for real-world applications the output of the forward model is typically obtained through numerical approximations of PDEs that require physical discretization, and therefore incorporates the model bias.  One window scoping the improvement of this approach that we aim to investigate is the combination of the multilevel Monte Carlo method, for the approximation of the tractile function, with high precision quadrature schemes for the tempering quadrature.

One intriguing  question coming up when observing the transition distributions is: Are the data goodness criteria, for parameter estimation problems,  following the same updating rate when the tempering parameter varies from $0$ to $1$? Depending on the answer, one would either reduce the method bias in the computation of the data goodness criteria by considering an earlier tempering point as the final stage. Or one can design a hierarchical clustering algorithm on a dataset for the inferential strength in parameter estimation problems. Future work will also explore this aspect for the Kullback-Leibler divergence in the framework of Bayesian experimental design.

%\backmatter

\section*{Acknowledgments}
The author is grateful to Kody J. Law (Manchester University), Youssef Marzouk (MIT, Cambridge) and Raul Tempone (KAUST \& Alexander von Humboldt Professor at RWTH Aachen University) for valuable discussions and comments.

%\section*{References}
%\nocite{*}

%\bibliographystyle{elsarticle-num}
\bibliographystyle{acm}
\bibliography{reference}

\newpage
\appendix

\section{Proof of Theorem \ref{theorem2}} \label{sec:app1}
The goal is to prove the existence and the uniqueness of solution for the Barbashin equation (\ref{eq:dif1}b). We introduce the abstract function $\hat{v}$ defined as
\begin{eqnarray}
  \hat{v} \colon
  \begin{aligned}[t]
        [0,1] &\longrightarrow \mathcal{S}_l\\
         \alpha &\longmapsto \hat{v}(\alpha) \myeq {v}(\alpha, \cdot).
  \end{aligned}
\end{eqnarray}
Therefore, we can rewrite (\ref{eq:dif1}b) as an abstract Cauchy problem
\begin{eqnarray}
\label{eq:abs1}
  \frac{\hbox{d} \hat{v}}{\hbox{d} \alpha} = A(\alpha) \hat{v}(\bm{\theta}),
\end{eqnarray}
where $\hat{v}(\bm{\theta}) = \hat{v}(\alpha)(\bm{\theta})$, the derivative $\frac{\hbox{d} \hat{v}}{\hbox{d} \alpha}$ is in the Fr\'echet sense, and the Cauchy operator $A$ is the integral operator given by
\begin{eqnarray}
   A(\alpha) \hat{v}(\bm{\theta}) = \int_\Theta \mathcal{H}(\alpha, \bm{\theta}, \bm{\eta}, \hat{v}(\bm{\eta})) \hbox{d}\bm{\eta}.
\end{eqnarray}
The classical solvability of \eqref{eq:abs1} follows from the strong continuity of the integral operator
\begin{eqnarray}
  A \colon
  \begin{aligned}[t]
        [0,1] &\longrightarrow \mathcal{L}(\mathcal{S})\\
         \alpha &\longmapsto A(\alpha),
  \end{aligned}
\end{eqnarray}
where $\mathcal{L}(\mathcal{S}_l)$ is the space of all bounded linear operators in the Banach space $\mathcal{S}_l$. Therefore we get the solvability since $A$ is strongly continuous in $\mathcal{L}(\mathcal{S}_l)$ by construction. Indeed, let $(\alpha_n)_{n>0}$ be a sequence of real valued numbers in $[0,1]$ that converges to $\alpha \in [0,1]$, 
% i.e.,
%\begin{eqnarray*}
%\forall \varepsilon >0, \quad \exists N\in \mathbb{N}\quad \hbox{such that} \quad \forall n\geq N \quad |\alpha_n - \alpha| < \varepsilon.
%\end{eqnarray*}
and $\|\cdot\|_{\mathcal{L}}$ be an appropriate norm defined in $\mathcal{L}(\mathcal{S}_l)$, we have
\begin{eqnarray*}
\left\| A(\alpha_n) - A(\alpha) \right\|_{\mathcal{L}} 
& = & | \int_\Theta    \mathcal{H}(\alpha_n, \bm{\theta}, \bm{\eta}, \hat{v}(\bm{\eta}))  -  \mathcal{H}(\alpha, \bm{\theta}, \bm{\eta}, \hat{v}(\bm{\eta})) \hbox{d} \bm{\eta} |  \\
& = & | \int_\Theta    \mathcal{H}(\alpha_n, \bm{\theta}, \bm{\eta}, \hat{v}(\alpha_n)(\bm{\eta}))  -  \mathcal{H}(\alpha, \bm{\theta}, \bm{\eta}, \hat{v}(\alpha)(\bm{\eta})) \hbox{d} \bm{\eta} |  \\ 
&\leq &  \int_\Theta  |  \mathcal{H}(\alpha_n, \bm{\theta}, \bm{\eta}, \hat{v}(\alpha_n)(\bm{\eta}))  -  \mathcal{H}(\alpha, \bm{\theta}, \bm{\eta}, \hat{v}(\alpha)(\bm{\eta})) | \hbox{d}\bm{\eta}  \\
& \leq & \int_\Theta  | \mathcal{H}(\alpha_n, \bm{\theta}, \bm{\eta}, \hat{v}(\alpha_n)(\bm{\eta}))  - \mathcal{H}(\alpha, \bm{\theta}, \bm{\eta}, \hat{v}(\alpha_n)(\bm{\eta})) | \hbox{d}\bm{\eta} \\
& &  +  \int_\Theta  | \mathcal{H}(\alpha, \bm{\theta}, \bm{\eta}, \hat{v}(\alpha_n)(\bm{\eta})) - \mathcal{H}(\alpha, \bm{\theta}, \bm{\eta}, \hat{v}(\alpha)(\bm{\eta})) | \hbox{d}\bm{\eta}
\end{eqnarray*}
%The tempering parameter affects the kernel $k$ implicitly through the dependance of the random variable sampled from the power posterior, $k(\alpha, \bm{\theta}, \bm{\eta}) \equiv k(\bm{\theta}, \bm{\eta})$. 
The functional $\hat{v}$ is continuous, and as  $\bm{g}$ is assumed continuous with respect to $\bm{\theta}$, the  kernel $\mathcal{H}$ is continuous. Therefore $A$ is strongly continuous. This completes the proof.

\section{Derivation of the kernel trace equations}\label{sec:app3}
We use the orthogonality condition \eqref{eq:orth} and the eigenvalue problem setting \eqref{eq:eigen} to characterize the dynamics of the eigenvalue $\lambda_n$.
\begin{eqnarray*}
\lambda_n(\alpha) &=& \lambda_n(\alpha) \left\langle k^\alpha_n(\bm{\theta}),k^\alpha_n (\bm{\theta})\right\rangle_{\alpha} \\
  & = & \lambda_n(\alpha)  \int_\Theta k^\alpha_n(\bm{\theta}) k^\alpha_n(\bm{\theta}) \pi(\bm{\theta}|\bm{Y};\alpha)\hbox{d} \bm{\theta}\\
 & = &  \int_\Theta k^\alpha_n(\bm{\theta}) \lambda_n(\alpha) k^\alpha_n(\bm{\theta}) \pi(\bm{\theta}|\bm{Y};\alpha) \hbox{d} \bm{\theta}\\
  & = &  \int_\Theta k^\alpha_n(\bm{\theta})\int_\Theta  \mathcal{K}(\alpha, \bm{\theta}, \bm{\eta})k^\alpha_n(\bm{\eta}) \pi(\bm{\eta}|\bm{Y};\alpha) \hbox{d}\bm{\eta} \;\pi(\bm{\theta}|\bm{Y};\alpha) \hbox{d} \bm{\theta}.
\end{eqnarray*}
Differentiation with respect to the parameter $\alpha$ gives
\begin{eqnarray*}
 \frac{\partial \lambda_n(\alpha)}{\partial \alpha} 
 &=& 
 \int_\Theta k^\alpha_n(\bm{\theta}) \int_\Theta  \mathcal{K}(\alpha, \bm{\theta}, \bm{\eta})k^\alpha_n(\bm{\eta})  \frac{\partial }{\partial \alpha} \pi(\bm{\eta}|\bm{Y};\alpha) \hbox{d} \bm{\eta} \;\pi(\bm{\theta}|\bm{Y};\alpha) \hbox{d} \bm{\theta}    \\
 & & + 
  \int_\Theta k^\alpha_n(\bm{\theta})\int_\Theta  \mathcal{K}(\alpha, \bm{\theta}, \bm{\eta})k^\alpha_n(\bm{\eta}) \pi(\bm{\eta}|\bm{Y};\alpha) \hbox{d} \bm{\eta} \; \frac{\partial }{\partial \alpha} \pi(\bm{\theta}|\bm{Y};\alpha) \hbox{d} \bm{\theta},
\end{eqnarray*}
yet from \eqref{eq:pide}, we obtain
\begin{eqnarray*}
\int_\Theta  \mathcal{K}(\alpha, \bm{\theta}, \bm{\eta})k^\alpha_n(\bm{\eta}) \frac{\partial }{\partial \alpha}  \pi(\bm{\eta}|\bm{Y};\alpha) \hbox{d} \bm{\eta} &=& \int_\Theta  \log p(\bm{Y}|\bm{\eta})  \mathcal{K}(\alpha, \bm{\theta}, \bm{\eta})k^\alpha_n(\bm{\eta})  \pi(\bm{\eta}|\bm{Y};\alpha) \hbox{d} \bm{\eta} \\
& & - \mathbb{E}_{\bm{\eta}|\bm{Y};\alpha} \left[\log p(\bm{Y}|\bm{\eta}) \right] \int_\Theta  \mathcal{K}(\alpha, \bm{\theta}, \bm{\eta})k^\alpha_n(\bm{\eta})  \pi(\bm{\eta}|\bm{Y};\alpha) \hbox{d} \bm{\eta}.
%
%&=& \Big( \log p(\bm{Y}|\bm{\theta}) - sL(\alpha) \Big) \int_\Theta  \mathcal{K}(\alpha, \bm{\theta}, \bm{\eta})k^\alpha_n(\bm{\eta})  \pi(\bm{\eta}|\bm{Y};\alpha) \hbox{d}\bm{\eta}.
\end{eqnarray*}
 It follows that
 \begin{eqnarray*}
  I_1 & \myeq  & \int_\Theta k^\alpha_n(\bm{\theta}) \int_\Theta  \mathcal{K}(\alpha, \bm{\theta}, \bm{\eta})k^\alpha_n(\bm{\eta})  \frac{\partial }{\partial \alpha} \pi(\bm{\eta}|\bm{Y};\alpha) \hbox{d} \bm{\eta} \;\pi(\bm{\theta}|\bm{Y};\alpha) \hbox{d} \bm{\theta} \\
  &=&   \int_\Theta k^\alpha_n(\bm{\theta}) \int_\Theta  \log p(\bm{Y}|\bm{\eta})  \mathcal{K}(\alpha, \bm{\theta}, \bm{\eta})k^\alpha_n(\bm{\eta})  \pi(\bm{\eta}|\bm{Y};\alpha) \hbox{d} \bm{\eta} \;\pi(\bm{\theta}|\bm{Y};\alpha) \hbox{d} \bm{\theta} \\
& & - \mathbb{E}_{\bm{\eta}|\bm{Y};\alpha} \left[\log p(\bm{Y}|\bm{\eta}) \right] \int_\Theta k^\alpha_n(\bm{\theta}) \int_\Theta  \mathcal{K}(\alpha, \bm{\theta}, \bm{\eta})k^\alpha_n(\bm{\eta})  \pi(\bm{\eta}|\bm{Y};\alpha) \hbox{d} \bm{\eta} \;\pi(\bm{\theta}|\bm{Y};\alpha) \hbox{d} \bm{\theta}.\\
  &=&   \int_\Theta k^\alpha_n(\bm{\theta}) \int_\Theta  \log p(\bm{Y}|\bm{\eta})  \sum_{l=1}^{\infty} \lambda_l(\alpha)k_l^\alpha(\bm{\theta})l_k^\alpha(\bm{\eta})k^\alpha_n(\bm{\eta})  \pi(\bm{\eta}|\bm{Y};\alpha) \hbox{d} \bm{\eta} \;\pi(\bm{\theta}|\bm{Y};\alpha) \hbox{d} \bm{\theta} \\
& & - \mathbb{E}_{\bm{\eta}|\bm{Y};\alpha} \left[\log p(\bm{Y}|\bm{\eta}) \right] \int_\Theta k^\alpha_n(\bm{\theta}) \lambda_n(\alpha) k^\alpha_n(\bm{\theta})  \;\pi(\bm{\theta}|\bm{Y};\alpha) \hbox{d} \bm{\theta}\\
  &=&   \sum_{l=1}^{\infty} \int_\Theta k^\alpha_n(\bm{\theta}) k_l^\alpha(\bm{\theta})\;\pi(\bm{\theta}|\bm{Y};\alpha) \hbox{d} \bm{\theta} \int_\Theta  \log p(\bm{Y}|\bm{\eta})  \lambda_l(\alpha)k_l^\alpha(\bm{\eta})k^\alpha_n(\bm{\eta})  \pi(\bm{\eta}|\bm{Y};\alpha) \hbox{d} \bm{\eta}  \\
& & - \lambda_n(\alpha)  \mathbb{E}_{\bm{\eta}|\bm{Y};\alpha} \left[\log p(\bm{Y}|\bm{\eta}) \right]\\
  &=&  \lambda_n(\alpha) \int_\Theta  \log p(\bm{Y}|\bm{\eta})  k_n^\alpha(\bm{\eta})k^\alpha_n(\bm{\eta})  \pi(\bm{\eta}|\bm{Y};\alpha) \hbox{d} \bm{\eta}  - \lambda_n(\alpha)  \mathbb{E}_{\bm{\eta}|\bm{Y};\alpha} \left[\log p(\bm{Y}|\bm{\eta}) \right].
 \end{eqnarray*}
 Similarly
 \begin{eqnarray*}
 I_2 & \myeq & \int_\Theta k^\alpha_n(\bm{\theta})\int_\Theta  \mathcal{K}(\alpha, \bm{\theta}, \bm{\eta})k^\alpha_n(\bm{\eta}) \pi(\bm{\eta}|\bm{Y};\alpha) \hbox{d} \bm{\eta} \; \frac{\partial }{\partial \alpha} \pi(\bm{\theta}|\bm{Y};\alpha) \hbox{d} \bm{\theta} \\
 & =&  \lambda_n(\alpha) \int_\Theta  \log p(\bm{Y}|\bm{\eta})  k_n^\alpha(\bm{\eta})k^\alpha_n(\bm{\eta})  \pi(\bm{\eta}|\bm{Y};\alpha) \hbox{d} \bm{\eta}  - \lambda_n(\alpha)  \mathbb{E}_{\bm{\eta}|\bm{Y};\alpha} \left[\log p(\bm{Y}|\bm{\eta}) \right].
 \end{eqnarray*}
 Therefore, the dynamics of an eigenvalue is
 \begin{eqnarray}
 \label{eq:eigendynamics1}
 \frac{\partial \lambda_n(\alpha)}{\partial \alpha}  = 2  \lambda_n(\alpha) \int_\Theta  \log p(\bm{Y}|\bm{\eta})  k_n^\alpha(\bm{\eta})k^\alpha_n(\bm{\eta})  \pi(\bm{\eta}|\bm{Y};\alpha) \hbox{d} \bm{\eta}  - 2 \lambda_n(\alpha)  \mathbb{E}_{\bm{\eta}|\bm{Y};\alpha} \left[\log p(\bm{Y}|\bm{\eta}) \right].
 \end{eqnarray}
Using
 \begin{eqnarray*}
  \log p(\bm{Y}|\bm{\eta}) = \log(\Cons) -  s \mathcal{K}(\alpha, \bm{\eta}, \bm{\eta})\quad \hbox{and}\quad  \mathbb{E}_{\bm{\eta}|\bm{Y};\alpha} \left[\log p(\bm{Y}|\bm{\eta}) \right] &=& \log(\Cons) -  s \int_\Theta  \mathcal{K}(\alpha, \bm{\eta}, \bm{\eta}) \pi(\bm{\theta}|\bm{Y};\alpha) \hbox{d}\bm{\eta} \\
  &=& \log(\Cons) -  s \sum_{n=1}^{\infty} \lambda_n(\alpha)  \\
  &=& \log(\Cons) -  s L_1(\alpha),
 \end{eqnarray*}
 we rewrite \eqref{eq:eigendynamics1} as 
\begin{eqnarray*}
 \frac{\partial \lambda_n(\alpha)}{\partial \alpha}  &=& 2s \lambda_n(\alpha)   L_1(\alpha) - 2s  \lambda_n(\alpha) \int_\Theta \mathcal{K}(\alpha, \bm{\eta}, \bm{\eta})  k_n^\alpha(\bm{\eta})k^\alpha_n(\bm{\eta})  \pi(\bm{\eta}|\bm{Y};\alpha) \hbox{d} \bm{\eta}\\
 &=&  2s\lambda_n(\alpha)   L_1(\alpha) - 2s  \int_\Theta \mathcal{K}(\alpha, \bm{\eta}, \bm{\eta})  \lambda_n(\alpha) k_n^\alpha(\bm{\eta})k^\alpha_n(\bm{\eta})  \pi(\bm{\eta}|\bm{Y};\alpha) \hbox{d} \bm{\eta},
\end{eqnarray*}
and summing over $n$, we obtain
\begin{eqnarray*}
  \frac{\partial  L_1(\alpha)}{\partial \alpha}  
  & =& 2s L_1^2(\alpha) - 2s  \int_\Theta \mathcal{K}(\alpha, \bm{\eta}, \bm{\eta}) \sum_{n=1}^{\infty} \lambda_n(\alpha) k_n^\alpha(\bm{\eta})k^\alpha_n(\bm{\eta})  \pi(\bm{\eta}|\bm{Y};\alpha) \hbox{d} \bm{\eta} \quad \hbox{as} \quad  L_1(\alpha) = \sum_{n=1}^{\infty} \lambda_n(\alpha)\\
    & =& 2s L_1^2(\alpha) - 2s  \int_\Theta \mathcal{K}^2(\alpha, \bm{\eta}, \bm{\eta})  \pi(\bm{\eta}|\bm{Y};\alpha) \hbox{d} \bm{\eta} 
\end{eqnarray*}
Thanks to the Schur product theorem, the Hadamard product $\mathcal{K}^2 = \mathcal{K} \circ \mathcal{K}$ of a positive definite kernel $\mathcal{K}$ is a symmetric positive definite kernel, and the eigenstructure expansion of $\mathcal{K}^2$, see \cite{Reade1981}, is

\begin{eqnarray}
  \mathcal{K}^2(\alpha, \bm{\eta}, \bm{\eta}) = \sum_{n=1}^{\infty} \lambda^2_n(\alpha) k_n^\alpha(\bm{\eta})k^\alpha_n(\bm{\eta}),
\end{eqnarray}
where $(\lambda_n, k^\alpha_n)$ are the eigenpair of the kernel $\mathcal{K}$. Therefore, we have
\begin{eqnarray*}
 \int_\Theta \mathcal{K}^2(\alpha, \bm{\eta}, \bm{\eta})  \pi(\bm{\eta}|\bm{Y};\alpha) \hbox{d} \bm{\eta} 
 & =&  \sum_{n=1}^{\infty} \lambda^2_n(\alpha) \int_\Theta    k_n^\alpha(\bm{\eta})k^\alpha_n(\bm{\eta}) \pi(\bm{\eta}|\bm{Y};\alpha) \hbox{d} \bm{\eta} \\
  & =&  \sum_{n=1}^{\infty} \lambda^2_n(\alpha)\\
  &=& L_2(\alpha),
\end{eqnarray*}
which implies that
\begin{eqnarray}
\label{eq:L1}
  \frac{\partial  L_1(\alpha)}{\partial \alpha}   = 2s\Big( L_1^2(\alpha) - L_2(\alpha) \Big).
\end{eqnarray}
Similar analysis, starting on differentiating the eigenvalue  $\lambda^2_n(\alpha)$ of $\mathcal{K}^2$, results in the dynamics of $L_2$ as
\begin{eqnarray}
  \label{eq:L2}
  \frac{\partial  L_2(\alpha)}{\partial \alpha}   = 2s\Big( L_1(\alpha)L_2(\alpha) - L_3(\alpha) \Big).
\end{eqnarray}
The skewness condition \eqref{eq:skewcondition} in terms of $L_i$, $i=1,2,3$ is
\begin{eqnarray}
\label{eq:skewL}
 - L_3(\alpha) = - 3 L_2(\alpha) L_1(\alpha) + 2 L_1^3(\alpha).
\end{eqnarray}
We gather together the equations \eqref{eq:L1}, \eqref{eq:L2} and \eqref{eq:skewL} to state the computational tractability of the expected deviance as
\begin{eqnarray*}
\left\{
\begin{array}{lll}
\displaystyle{\frac{\partial  L_1(\alpha)}{\partial \alpha}   = 2s\Big( L_1^2(\alpha) - L_2(\alpha) \Big),}&  \\\\
\displaystyle{ \frac{\partial  L_2(\alpha)}{\partial \alpha}   = 2sL_1(\alpha)\Big( L_1^2(\alpha) - L_2(\alpha) \Big).}&  
\end{array}
\right.
\end{eqnarray*}
This completes the derivation.

\end{document}